\documentclass[11pt]{article}
\usepackage{amsmath,amsfonts,amsthm,amssymb,verbatim}
\usepackage{epsfig}

\numberwithin{equation}{section}

\newcommand{\beq}{\begin{equation}}
\newcommand{\rd}{{\rm d}}
\newcommand{\ann}{a^{\hphantom{+}}}
\newcommand{\cre}{a^{\dagger}}
\newcommand{\wL}{{\widetilde L}}

\newcommand{\wh}{\widehat}

\newcommand{\eeq}{\end{equation}}
\newcommand{\beqa}{\begin{eqnarray}}
\newcommand{\eeqa}{\end{eqnarray}}
\newcommand\Z{{\mathbb Z}}
\newcommand\N{{\mathbb N}}

\newcommand\R{{\mathbb R}}
\newcommand\eps{\varepsilon}

\newcommand\rhov\varrho

\newcommand\rmor{ \,\,\,{\rm or}\,\,\, }
\newcommand\rmfor{ \,\,\,{\rm for}\,\,\, }
\newcommand\rmand{ \,\,\,{\rm and}\,\,\, }

\newcommand\const{{\rm const.\, }}
\newcommand\Tr{{\rm Tr}}

\newcommand\al{{\alpha}}

\newcommand\be{\beta}

\newcommand\x{{\rm x}}

\newcommand\rhoc{\rho_{\,c}}

\renewcommand\kappa\varkappa
\renewcommand\rho\varrho

\newtheorem{thm}{THEOREM}
\newtheorem{prop}{Proposition }
\newtheorem{lem}{Lemma}
\newtheorem{cor}{COROLLARY}
\newtheorem{mydef}{DEFINITION}
\theoremstyle{definition}

\begin{document}

\title{\bf{  Free Energies of Dilute Bose gases: Upper bound}}
\author{Jun Yin }
\maketitle

\begin{abstract}
We derive an upper bound on the free energy of a Bose gas at density $\rho$ and temperature $T$. In combination with the lower bound derived previously by
Seiringer \cite{RS1}, our result proves that in the low density limit, i.e., when $a^3\rho\ll 1$, where $a$ denotes the scattering length of the
pair-interaction potential,  the leading term of  $\Delta f$, the free energy difference per volume between interacting and ideal Bose gases, is equal to $4\pi a (2\rho^2-[\rho-\rhoc]^2_+)$. Here, $\rhoc(T)$ denotes the critical density for Bose-Einstein condensation
(for the ideal Bose gas), and $[\cdot ]_+$ $=$ $\max\{ \cdot , 0\}$ denotes the positive
part. 
\end{abstract}

\newpage 
 
\section{Introduction}
The ground state energy and the free energy are the fundamental properties of a quantum system and 
they have been intensively studied since the invention of the quantum mechanics. 
The recent progresses  in experiments on 
 Bose-Einstein condensation, especially the achievement of Bose-Einstein condensation in dilute gases of alkali atoms in 1995 \cite{exBCE}, have inspired re-examination of the theoretic foundation  
 concerning the Bose system, e.g., \cite{LY}, \cite{LS1}, \cite{LS2}, \cite{JOL}, \cite{JOLY} \cite{ESY}, \cite{Y}, \cite{GS} and \cite{YY} on ground state energy  and \cite{RS1} on free energy.    
 \par In the low density limit, the leading term of the ground state energy per volume 
 was identified  rigorously by   Dyson (upper
bound) \cite{D} and Lieb-Yngvason (lower bound) \cite{LY} to be $4\pi a \varrho^2$, 
where $a$ is the scattering length of the two-body
potential and $\varrho$ is the density. We note that $4\pi a\rho^2$ is also the first leading term of $\Delta E$, the ground state energy difference per volume between interacting and ideal Bose gases.(The ground state energy per volume of the ideal Bose gas is zero).  
\par On the other hand, the first leading term of $\Delta f$,  the free energy difference  between interacting and ideal Bose gases, is the second leading order term of the free energy per volume $f$. More specifically, if $a^3\rho\ll 1$, where $a$ denotes the scattering length of the
pair-interaction potential,  then  
\beq\label{int1}
f(\rho, T)=f_0(\rho,T)+4\pi a (2\rho^2-[\rho-\rhoc]^2_+)+o(a\rho^2)
\eeq 
Here, $f$ is the free energy per volume of the interacting Bose gas, $f_0$ is the one of the ideal Bose gas, $\rhoc(T)$ denotes the critical density for Bose-Einstein condensation (for the ideal gas), and $[\cdot ]_+$ $=$ $\max\{ \cdot , 0\}$ denotes the positive
part.  The lower bound on $f$ has been proved in Seiringer's work \cite{RS1}
. In this paper, we prove the upper bound on $f$ and obtain the main result \eqref{int1} 
\par The trial state  we use in this proof is of a new type, which was  first used in   \cite{YY}
. Let $\phi_0$ be the ground state of the ideal Bose gas.  In \cite{YY}, we constructed a trial state (pure state) for interacting Bose gases which is obtained by slightly modifying a state of the following form,   
\beq\label{trialstate00}
 \exp \Big  [   \sum_{k\sim 1} \sum_{v \sim \sqrt \rho }  2 \sqrt { \lambda_{k+v/2}  \lambda_{-k+v/2}} \cre_{k+v/2}\cre_{-k+v/2}\ann_v\ann_0    +   \sum_k  {c_k}  \cre_k\cre_{-k}\ann_0\ann_0   \Big  ] |\phi_0\rangle,
\eeq
(with suitably chosen $c$ and $\lambda$). Here the notation $ A\sim B$ means that $A$ and $B$ have the same order. The expression of \eqref{trialstate00} is simple but it is hard to use itself for our calculation in   \cite{YY}. If one tried to write \eqref{trialstate00} with  the occupation-number representation as (for  calculting interaction energies)
\beq
\sum _{\al }f_\al|\al\rangle,
\eeq
he will see that it is very hard to  calculate $f_\al$'s. Therefore in  \cite{YY}, we constructed a trial state $\sum _{\al }\tilde f_\al|\al\rangle$ by defining $\tilde f_\al$ directly. The $\tilde f_\al$'s have many properties, which have no physical meaning but can simplify our proof. E.g. if the state $|\al\rangle $ contains a particle with extremely high momentum, then $\tilde f_\al=0$.   Furthermore, the trial state $\sum _{\al }\tilde f_\al|\al\rangle$ is very close to \eqref{trialstate00} i.e., for some $c>0$,
\beq
\sum _{\al }|f_\al-\tilde f_\al|^2 \langle \al|\al\rangle\ll \rho^{c}.
\eeq 
This basic idea will be used again in this paper. 

\par This trial state (pure state) in \cite{YY} is used to rigorously prove the upper bound of the second order
correction to the ground state energy, which  was first computed by Lee-Yang \cite{LYang} (see also Lee-Huang-Yang \cite{LHY} and  the recent paper 
by Yang \cite{RY} for results in other dimensions.  Another derivation 
was later given  by Lieb \cite{L} using 
a self-consistent closure assumption  for the hierarchy of correlation functions.)
\par We can rewrite the pure state \eqref{trialstate00} as follows 
\beq
\eqref{trialstate00}=P_{(0,0)} P_{(0,\sqrt\rho)}|\phi_0\rangle
\eeq 
where
\beqa
P_{(0,0)}&&=\exp\left[\sum_{k\sim 1}  {c_k}  \cre_k\cre_{-k}\ann_0\ann_0\right]\\\nonumber
P_{(0,\sqrt\rho)}&& = \exp\left[\sum_{k\sim 1} \sum_{v \sim \sqrt \rho }2 \sqrt { \lambda_{k+v/2}  \lambda_{-k+v/2}} \cre_{k+v/2}\cre_{-k+v/2}\ann_v\ann_0\right]
\eeqa
We note: $P_{(0,0)}$ represents the interactions between condensate and condensate, since in the operator $\cre_k\cre_{-k}\ann_0\ann_0$ two particles with momenta zero are annihilated $(\ann_0\ann_0)$ and two particles with high momentum are  created $(  \cre_k\cre_{-k})$.  Similarly $P_{(0,\sqrt\rho)}$ represents the interaction between condensate and the particles with momentum  of order $\rho^{1/2}$, since in this operator one particle with momentum zero and one with momentum of order $\rho^{1/2}$ are annihilated $(\ann_v\ann_0)$ and two particles with high momenta are created.
\par In this paper, we construct a trial state of a similar form. More specifically, let $\Gamma_I$ be Gibbs state of the ideal Bose gas at temperature $T$. The trial state we are going to use is very close to
\beq\label{fakets}
\Gamma\sim\left(P_{(\rho^{1/3},\rho^{1/3})}P_{(0,\rho^{1/3})}P_{(0,0)}\right)\Gamma_I \left(P_{(\rho^{1/3},\rho^{1/3})}P_{(0,\rho^{1/3})}P_{(0,0)}\right)^\dagger
\eeq
where 
\beqa\label{defP0rho}
P_{(0,0)}&&=\exp\left[\sum_{k\sim 1}  {c_k}  \cre_k\cre_{-k}\ann_0\ann_0\right]\\\nonumber
P_{(0,\rho^{1/3})}&& = \exp\left[\sum_{k\sim 1} \sum_{v \sim \rho^{1/3} }2 \sqrt { \lambda_{k+v/2}  \lambda_{-k+v/2}} \cre_{k+v/2}\cre_{-k+v/2}\ann_v\ann_0\right]\\\nonumber
P_{(\rho^{1/3},\rho^{1/3})}&& = \exp\left[\sum_{k\sim 1} \sum_{u\neq v \sim \rho^{1/3} } \sqrt { \lambda_{k+\frac{v+u}2}  \lambda_{-k+\frac{v+u}2}} \cre_{k+\frac{v+u}2}\cre_{-k+\frac{v+u}2}\ann_v\ann_u\right]
\eeqa
where the constant $2$ comes from the ordering of 
$\ann_v \ann_0$. As one can see, $P_{(0,0)}$ represents the interactions between condensate and condensate, $P_{(0,\rho^{1/3})}$ represents the interaction between condensate and the particles with momentum of order $\rho^{1/3}$, and $P_{(\rho^{1/3},\rho^{1/3})}$ represents the interaction  between the particles with momentum of order $\rho^{1/3}$.  

\section {Model and Main results}
\subsection{Hamiltonian and Notations}
We consider a Bose gas which is composed of $N$ identical bosons confined to a cubic box $\Lambda$ of side length $L$. 
The Hilbert space $\mathcal H_{N,\,\Lambda}$ for the system is the set of symmetric functions in $L^2(\Lambda^N)$. The Hamiltonian is given as
\beq\label{defH}
H_{N,\Lambda}=-\sum_{i=1}^N\Delta_i+\sum_{1\leq i< j\leq N}V(x_i-x_j)
\eeq
Here $x_i\in \Lambda$ ($1\leq i\leq N$) is the position of $i$th particle.   The two body interaction is given by  a spherically symmetric  non-negative function $V$, such that $\|V\|_\infty<\infty$, as in \cite{YY} and \cite{ESY}.    In the proof on the lower bound of the free energy, \cite{RS1}, the $V$ is assumed to have a finite range $R_0$, i.e., $V(r) = 0$ for $r > R_0$. Therefore we will also use  this assumption in this paper.  In particular, it has a finite
scattering length, which we denote by $a$. 
\par  We note that the interaction only depends on the distance between the particles.  
As usually, we denote by  $H_{N,\,\Lambda}^P$ ($H_{N,\,\Lambda}^D$) the Hamiltonians with periodic 
(Dirichlet) boundary conditions (Here $x_i-x_j$ in \eqref{defH} is really the distance on the torus in the periodic case). 
\par In periodic case, we can also write Hamiltonian with creation and annihilation operators as follows. The dual space of $\Lambda$ is 
$\Lambda^*:= (\frac{2\pi}{L} \Z)^3$. 
For a continuous function $F$ on $\R^3$, we have
\beq
    \frac{1}{L^3} \sum_{p\in\Lambda^*} F(p) =
   \frac{1}{|\Lambda|} \sum_{p\in\Lambda^*} F(p)
  \stackrel{|\Lambda|\to \infty}{\longrightarrow} \int_{\R^3} 
\frac{{\rm d}^3 p}{(2\pi)^3} F(p)
\eeq

  The Fourier
transform is defined as
$$
     \wh V_p = \int_\Lambda e^{-ipx} V(x) \rd x, \qquad
     V(x) = \frac{1}{|\Lambda|} \sum_{p\in \Lambda^*} e^{ipx} \wh V_p $$ and then $$
     \frac{1}{|\Lambda|} \sum_{p\in \Lambda^*} e^{ipx} = \delta_{\R^3}(x), \qquad \int_\Lambda e^{ipx}\rd x = \delta_{\Lambda^*}(p) $$ where $\delta_{\R^3}(x)$ is the usual continuum delta function and the function 
$\delta_{\Lambda^*}(p)=
|\Lambda|= L^3$ if $p=0$ (otherwise it is zero) is the lattice delta-function.
We will neglect the subscript; the argument indicates whether it is the momentum or position space delta function. In general 
we will also neglect the hat in the Fourier
transform. To avoid confusion, we follow the convention that the variables $x,y,z$ etc denote position space, the variables $p,q,k,u,v$ etc.\ denote  momentum space. 
We also simplify the notation
$$
\sum_p:=   \sum_{p\in \Lambda^*}
$$
i.e.\ momentum summation is always over  $\Lambda^*$.
We will use the bosonic operators with the commutator relations
$$
      [\ann_p, \cre_q] = \ann_p \cre_q -\cre_q\ann_p
     = \left\{
\begin{array}{ll}
  1 & \mbox{ if } p=q \\ 
0 & \mbox{ otherwise.}
\end{array} 
\right.
$$
Thus our   Hamiltonian in the Fock space $\mathcal F_\Lambda=\oplus_{N}\mathcal H_{N,\Lambda}$ is given by
\beq
      H_\Lambda ^P= 
\sum_p p^2 \cre_p\ann_p + \frac{1}{|\Lambda|} \sum_{p,q,u}
     \frac{\wh V_u}2 \cre_p\cre_q\ann_{p-u}\ann_{q+u},
\label{ham}
\eeq
\subsection{Free energy}
The free energy per unit volume of the system at  temperature
$ T=\beta^{-1} > 0$ and density $\rho=N/|\Lambda|> 0$ in the cubic box $\Lambda$ is defined as
\beq\label{deff}
f(\rho\,, \Lambda, \beta) \equiv -\frac1{|\Lambda|\beta}\ln\left(\Tr_ {\mathcal H_{N,\Lambda}} Exp (-\beta H_{N,\,\Lambda})\right),
\eeq
Let $f^P(\rho\,, \Lambda, \beta)$ and $f^D(\rho\,, \Lambda, \beta)$ denote the free energy per unit volume of the system with periodic or Dirichlet boundary conditions.
Furthermore, we denote by $f(\rho,\beta)$ the free energy (per unit volume) in the thermodynamic limit, i.e., $|\Lambda|$, $N\to \infty$ with
$\rho= N/|\Lambda|$ fixed, i.e., 
\beq\label{deffth}
f^{P(D)}(\rho\,, \beta)\equiv \lim_{|\Lambda|\to\infty}f^{P(D)}(\rho\,, \Lambda, \beta)
\eeq
As mentioned in the introduction, in this paper we give an upper bound on the leading order correction of $f(\rho\,,\beta)$, compared with an
ideal gas, in the case that $a^3\rho$ is small and  $\beta\rho^{2/3}$ is order one. We note that $a^3\rho$ and $\beta\rho^{2/3}$ are dimensionless quantities.
\subsection {Ideal Bose gas in the Thermodynamic Limit}
In this section, we review some well known results on ideal Bose gases. In the case of vanishing interaction potential ($V=0$), the free energy per unit volume in the thermodynamic limit can be
evaluated explicitly. Let $\zeta$ denote the Riemann zeta function.  It is well known that when $\rho^{2/3}\beta\geq (4\pi)^{-1}\zeta(3/2)^{2/3}$, i.e., $\rho$ is greater than critical density $\rho_c$,
\beq\label{rhoc}
\rho\geq\rho_{\,c}\equiv  (4\pi\beta)^{-3/2}\zeta(3/2)
\eeq
the free energy in the thermodynamic limit is given as 
\beq\label{deff01}
f^{D(P)}_0(\rho, \beta) =\frac{1}{(2\pi)^3\beta}\int_{\R^3}\ln(1-e^{-\beta p^2})d^3p
\eeq
\par On the other hand, when $\rho\leq \rhoc$,
\beq\label{deff02}
f^{D(P)}_0(\rho, \beta) =\rho\, \mu +\frac{1}{(2\pi)^3\beta}\int_{\R^3}\ln(1-e^{-\beta (p^2-\mu)})d^3p
\eeq 
Here $\mu(\rho,\beta)<0$ is determined by
\beq\label{defmu}
\rho=\frac{1}{(2\pi)^3}\int_{\R^3}\frac{1}{e^{\beta (p^2-\mu)}-1}d^{\,3}p
\eeq
Note: when $\rho\geq\rho_c$, $\mu(\rho,\be)$ is defined as zero. 
\par It is easy to see  the scaling relation: $$f^{D(P)}_0(\rho, \beta) = \rho^{5/3}f^{D(P)}_0( 1, \rho^{2/3}\beta)$$ and the ration $\rhoc/\rho$ only depends on dimensionless quantity $\rho^{2/3}\beta$, i.e., 
\beq
\rhoc/\rho=(4\pi)^{-3/2}\zeta(3/2)(\rho^{2/3}\beta)^{-3/2}
\eeq 
Let $\beta(\rho)$ be a function of  $\rho$, we define $R[\beta]$ as the ratio $\rhoc/\rho$ in the limit $\rho\to0$, i.e., 
\beq\label{defRbeta}
R[\beta]\equiv\lim_{\rho\to 0}\rho_c(\beta)/\rho=\lim_{\rho\to 0}(4\pi)^{-3/2}\zeta(3/2)\left(\rho^{2/3}\beta(\rho)\right)^{-3/2}
\eeq
\subsection{Scattering length}
In this paper, we use the standard definition of scattering length, as in   \cite{LY}, \cite{JOL}, \cite{ESY}, \cite{Y}, \cite{GS}, \cite{YY}, \cite{RS1}. Let $1-w$ be the zero energy scattering solution, i.e., 
\beq\label{zess}
    -\Delta(1-w) +  \frac12 V(1-w)=0
\eeq
with $0\leq w<1$ and $w(x)\to 0$ as $|x|\to \infty$.
Then the scattering length is given by the formula 
\beq\label{defa}
   a: =\frac{1}{4\pi}\int_{\R^3} \frac12 V(x)(1-w(x))\rd x
\eeq
With \eqref{zess}, we have, for $p\neq 0$, 
\beq\label{defw}
w_p=\left[\frac12 V(1-w)\right]_p|p|^{-2},
\eeq
Because $V(1-w)\geq 0$, so for $\forall p$, 
$$\left|\left[V(1-w)\right]_p\right|\leq \int V(1-w).$$
 Then with \eqref{defa}, i.e., 
$\int \frac12 V(1-w)$ is equal to $4\pi a$,  we obtain the following bound on $w_p$
\beq\label{boundwp}
\left|w_p\right|\leq 4\pi a |p|^{-2}
\eeq
Furthermore, when $V$ is $C^\infty$ function with compact support, one can easily prove that  
\beq\label{derwp}
\left|\frac{dw_p}{dp}\right|\leq \const \left(|p|^{-3} +|p|^{-2}\right)
\eeq
Here the constant only depends on $a$ and $R_0$. 
\subsection{Main results}
\begin{thm}\label{mainthm}  Let $V(x)\geq 0$ be a bounded, piecewise countinous function with compact support.   In the  temperature region where  $\lim_{\rho\to 0}\rho^{2/3}\beta(\rho)\in(0,\infty)$ and in the thermodynamic limit,  we have the following upper bound on  the free energy difference per volume between the interacting Bose gas $f^{D}(\rho, \beta)$ and the ideal Bose gas $f_0^{D}(\rho, \beta)$:
\beq\label{main1}
\overline\lim_{\rho\to 0}\left(f^{D}(\rho, \beta)-f^{D}_{\,0}(\rho,\beta)\right)\rho^{-2}\leq 4\pi a(2-[1-R[\beta]\,]^2_+),
\eeq
where $R[\beta]$ is defined in \eqref{defRbeta} as the ratio $\rhoc/\rho$ in the limit $\rho\to0$, and $a$ is the scattering length of $V$. 
\end{thm}

\bigskip 
\par It is well known that the effect of boundary conditions for free particles in the thermodynamic limit is negligible, i.e., 
\beq
f_0(\rho, \beta)\equiv f_0^{D}(\rho, \beta)=f_0^{P}(\rho, \beta)=f_0^{N}(\rho, \beta)=f_0^{R}(\rho, \beta)
\eeq
where $N$ denotes Neumann conditon and $R$ denotes Robin boundary condition: $\partial u/\partial \nu= -\alpha u$ (for some given constant $\alpha> 0$, with $\nu$
denoting the outward normal). 

On the other hand,  the proposition 2.3.5 and 2.3.7 of \cite{Rob} show that   
\beq
f^{D}(\rho, \beta)=f^{P}(\rho, \beta)=f^{N}(\rho, \beta)=f^{R}(\rho, \beta).
\eeq
 Therefore, with the results on lower bound  in Seiringer's work \cite{RS1}, we can obtain the following result.
\begin{cor} Under the assumption of Theorem \ref{mainthm}, in Dirichlet, perodic, Neumann and Robin boundary condition,  we have:
\beq\label{main2}
\lim_{\rho\to 0}\left(f^{P(N, D,R)}(\rho, \beta)-f_{\,0}(\rho,\beta)\right)\rho^{-2}= 4\pi a(2-[1-R[\beta]\,]^2_+),
\eeq
\end{cor}
\section{Basic strategy }
\subsection{Reduction to Small Torus with Periodic Boundary Conditions}
\par To obtain the upper bound to the free energy, we can use the variational principle, which states that, for any state  $\Gamma^{D(P)}$ ($\mathcal H_N\to \mathcal H_N$) in the domain of $H_{N,\Lambda}^{D(P)}$ (we will omit these superscripts of $H$ since it will be clear from the context
what they are), the following inequality holds.
\begin{equation}\label{varp}
f^{D(P)}(\rho,\,\Lambda,\, \beta)\leq \frac1{|\Lambda|} \Tr_{{\cal H}_{N,\Lambda}}\, H_{N,\Lambda} \Gamma^{D(P)} - \frac
1{|\Lambda|\beta} S(\Gamma^{D(P)})
\end{equation}
Here, $S(\Gamma)=-\Tr\,
\Gamma\ln\Gamma$ denotes the von Neumann entropy. Hence, to prove Theorem \ref{mainthm}, one only needs to construct a trial states $\Gamma^{D}(\rho,\, \Lambda,\, \beta)$  satisfying   Dirichlet boundary condition and the following inequality:
\beqa\nonumber
&&\overline\lim_{\rho\to 0}\overline\lim_{|\Lambda|\to \infty}\left(\frac1{|\Lambda|}\Tr\, H_{N,\Lambda} \Gamma^{D} - \frac
1{|\Lambda|\beta} S(\Gamma^{D})-f^{D}_{\,0}(\rho,\beta)\right)\rho^{-2}\\
\leq&& 4\pi a(2-[1-R[\beta]\,]^2_+)
\eeqa 
Furthermore, the proper trial states in the thermodynamic limit ($\Lambda\to\infty$) can be constructed by duplicating the proper trial states in the \textit{small} boxes ($|\Lambda|=\rho^{-c}, c>2$) with Dirichlet boundary condition. (Let the distance between the adjacent small boxes be $R_0$. Therefore there is no interaction between different boxes.)   Hence, the following Proposition  \ref{mainprop} implies our main result, Theorem \ref{mainthm}.

Note: Late we will choose the volume of the small box as $\rho^{-2-\eps}$, where $\eps$ is a small positive number. As one can see that, when size of  the box is too small, the Dirichilet Boundary condition will affect (increase)  the (total) free energy. When the volume of the small box is $O(\rho^{-2})$, we noticed that we can not prove that the effect of   Dirichilet Boundary condition is much less than the effect of the interaction. Therefore, to study the effect of the interaction, we have to choose the volume of the small box as $\rho^{-2-\eps}$.

\begin{prop}\label{mainprop}
In the  temperature region where  $\lim_{\rho\to 0}\rho^{2/3}\beta(\rho)\in (0,\infty)$, for fixed scattering length $a$,  there exist $  \Lambda$ with   $|  \Lambda|\geq \rho^{\,-41/20}$ and  trial states $\Gamma^{D}( \rho,\,   \Lambda,\,  \beta)$ satisfying the Dirichlet boundary condition and the inequality (set $ N=  |\Lambda| \rho$)
\beqa\nonumber
&&\overline\lim_{ \rho\to 0}\left(\frac1{| \Lambda|}\Tr\, H_{N,\Lambda} \Gamma^{D} - \frac
1{| \Lambda|\beta} S(\Gamma^{D})-f^{D}_{\,0}( \rho, \beta)\right)\rho^{-2}\\\label{mainpropresult}
\leq&& 4\pi a(2-[1-R( \beta)\,]^2_+),
\eeqa
where $R( \beta)$ is defined  in \eqref{defRbeta}.
\end{prop}
\par Here the number $41/20$ in the assumption can be replaced with any number larger than 2. 
\par On the other hand, the next lemma shows that a Dirichlet boundary condition trial state with correct free energy can be obtained from a   
periodic trial state in a slightly smaller box. 

\begin{lem}\label{relationDP}
Let the volume $|\Lambda|$  be equal to $\rho^{-41/20}$. In the temperature region of theorem \ref{mainthm}, if  
\beq\label{roughfprlb}
f^{P}(\rho,\, \Lambda,\,\beta)\leq \const \rho^{5/3},
\eeq
then for the revised box $\Lambda^ *$ and density $\rho^*$, defined by
\beq\label{relationtilde}
| \Lambda^ *|\equiv |\Lambda|(1+2\rho^{41/120})^3,\,\,\, \rho^*\equiv\rho(1+2\rho^{41/120})^{-3},\,\,\,
\eeq we have $f^{D}(\rho^*,\, \Lambda^ *,\,\beta)$ bounded from above as follows
\beqa\label{temp3.6}
&&\overline\lim_{\rho\to 0}\left(f^{D}(\rho^*,\, \Lambda^ *,\,\beta)-f^{P}(\rho,\,\Lambda, \beta)\right)\rho^{-2}\leq  0
\eeqa
\end{lem}
Lemma \ref{relationDP} can be proved with standard methods as in \cite{YY} and we postpone the proof to section  \ref{proofrelationDP}. 

We note: $| \Lambda^ *|\geq(\rho^*)^{-41/20}$, and satisfies the assumption in Proposition \ref{mainprop}. The construction of  a periodic trial state yielding the correct free energy upper bound is the core of this paper. 
We state it as the following theorem, which  gives the upper bound on $f^{P}(\rho,\, \Lambda,\,\beta)$ in \eqref{roughfprlb} and \eqref{temp3.6}.
\begin{thm}\label{resultP}
Assume $\lim_{\rho\to 0}\rho^{2/3}\beta\in (0,\infty)$ . For $|\Lambda|=\rho^{-41/20}$ and $N=|\Lambda|\rho$, there exists a periodic trial state $\Gamma(
\rho, \Lambda,\be)$ satisfying
\beq\label{resultPresult1}
\overline\lim_{\rho\to 0}\left(\frac1{|\Lambda|}\Tr\, H_N \Gamma - \frac
1{|\Lambda|\beta} S(\Gamma)-f^{P}_{\,0}(\rho,\beta)\right)\rho^{-2}
\leq 4\pi a(2-[1-R[\beta]\,]^2_+)
\eeq 
It implies
\beq\label{resultPresult2}
\overline\lim_{\rho\to 0}\left(f^{P}(\rho,\,\Lambda, \beta)-f^{P}_{\,0}(\rho,\,\beta)\right)\rho^{-2}\leq  4\pi a(2-[1-R[\beta]\,]^2_+)
\eeq
\end{thm}

\subsection{Proof of  Proposition \ref{mainprop}}
\par To prove  Proposition \ref{mainprop}, we can directly apply Lemma \ref{relationDP} and Theorem \ref{resultP}.  Lemma \ref{relationDP} shows that the upper bound of the free energy (with Dirichlet boundary conditions) is sightly larger than  the one (with Periodic boundary conditions) in a slightly smaller box. In the smaller box the density is sightly increased. But the temperature is unchanged. Therefore the relation between temperature and density is different from the one in the initial small box. In this subsection, we will show that this difference will not affect our result(up to the order $\rho^2$).
\par {\it Proof of Proposition \ref{mainprop}} 
\par Using the temperature function $ \beta$ in the assumption of  Proposition \ref{mainprop},  we define a new temperature function $\widetilde\beta$ as follows 
\beq\label{defwbeta}
\widetilde \beta: \,\widetilde \beta(\rho)= \beta(\rho^*),
\eeq
where $\rho^*=\rho(1+2\rho^{41/120})^{-3}$, as  in \eqref{relationtilde}. 
\par  Insert the result in Theorem \ref{resultP} into Lemma \ref{relationDP}.  
With the definition of $\Lambda^*$, $\rho^*$ in Lemma \ref{relationDP}\eqref{relationtilde}, we obtain at the inverse temperature $\widetilde \beta(\rho)$,
\beq\label{temp3.10}
\overline\lim_{\rho\to 0}\left(f^{D}( \rho^*,\,  \Lambda^*,\, \widetilde \beta)-f^{P}_{\,0}(\rho,\,\widetilde\beta)\right)\rho^{-2}\leq  4\pi a(2-[1-R[\widetilde\beta]\,]^2_+).
\eeq
Since $\rho^*=\rho(1+o(\rho^{1/3}))$, we have the following equalities on  the free  energies of  ideal Bose gases in the thermodynamic limit: 
\beq
f^{P}_{\,0}( \rho,\, \widetilde\beta)=f^{D}_{\,0}( \rho,\, \widetilde\beta)=f^{D}_{\,0}( \rho^*,\, \widetilde\beta)(1+o(\rho^{1/3})).
\eeq
Therefore, we can replace $f^{P}_{\,0}(\rho,\,\widetilde\beta)$ in \eqref{temp3.10} with  $f^{D}_{\,0}(\rho^*,\,\widetilde\beta)$, i.e.,
\beq
\overline\lim_{\rho\to 0}\left(f^{D}( \rho^*,\,  \Lambda^*,\, \widetilde\beta)-f^{D}_{\,0}(\rho^*,\,\widetilde\beta)\right)\rho^{-2}\leq  4\pi a(2-[1-R[\widetilde\beta]\,]^2_+).
\eeq
Then by the definition of $\widetilde\beta$ in \eqref{defwbeta}, we obtain $R[\beta]=R[\widetilde\beta]$, so
\beqa\nonumber
\overline\lim_{\rho\to 0}\left(f^{D}( \rho^*,\,  \Lambda^*,\,  \beta(\rho^*))-f^{D}_{\,0}( \rho^*,\,\beta(\rho^*))\right)\rho^{-2}\leq&&  4\pi a(2-[1-R[\beta]\,]^2_+)\\\nonumber
=&&4\pi a(2-[1-R[\widetilde\beta]\,]^2_+)
\eeqa
Finally, using that  $\Lambda^*\geq (\rho^*)^{-\frac{41}{20}}$ and 
the fact that the limit $\rho\to 0$ is equivalent to the limit $
\rho^*\to 0$, we arrive at the desired result \eqref{mainpropresult}. 
\qed
\subsection{Outline of the Proof of Theorem \ref{resultP}: Reduction to Pure States}
 \par As we showed in appendix, for any non-negative, bound, piecewise continous, spherically symmetric function $f$ supported in unit ball, 
 there exist $C^\infty$ non-negative, spherically symmetric function $f_{1}$, $f_2, \ldots$ supported in the ball of radiu 2, such that for any $i\geq 1$, 
 \beq
 f_i-f\geq 0\,\,\,and\,\,\, \lim_{i\to \infty}\|f_i-f\|_1\to 0
 \eeq

 Therefore, for any $\eps>0$, there exists a $C^\infty$ function $V^\eps$ with compact support such that $V^\eps\geq V$ and the scattering length of $V^\eps$ is less than $a+\eps$. By the definition of free energy and the variational principle,  
 \beq
 f(\rho,\beta, \Lambda)\leq f^\eps(\rho,\beta, \Lambda)
 \eeq
 where $f^\eps$ corresponds to the Bose gas with interaction $V^\eps$. Therefore  to prove Theorem \ref{resultP} and \eqref{resultPresult1}, we only need to  focus on the $V$'s that are $C^\infty $- functions and have compact support. 
 Hence in the remainder of this paper we assume that $V$ is $C^\infty$.
 \par  
 In this subsection, we introduce the basic strategy of  proving  Theorem  \ref{resultP}. With the assumption of Theorem \ref{resultP}, we have \beq\Lambda =[0, L]^3,\;L=\rho^{-\frac{41}{60}},\;N=\rho^{-\frac{21}{20}}\rmand\lim_{\rho\to0}\rho^{2/3}\beta\in(0,\infty).
\eeq 
\par We first identify  four regions in the momentum space  $\Lambda^*$  which are  relevant to the construction of the trial state: $P_0$ for the condensate;  $P_L$ for the low momenta, which are of the order  $\rho^{1/3}$;  $P_H$ for momenta of order one; and $P_I$ the region between 
$P_0$ and $P_L$. 
\begin{mydef} \label{def1} {Definitions of $P_0$, $P_I$, $P_L$ and $P_H$}
\par Define four subsets of momentum space $\Lambda^*=(2\pi L^{-1}\Z)^3$: $P_0$, $P_I$, $P_L$ and $P_H$ as follows.
\beqa\nonumber
P_0&&\equiv\left\{p=0\right\}\\\nonumber
P_I&&\equiv\left\{p\in \Lambda^*:0 <   |p|< \eps_L\rho^{1/3}\right\}\\\label{defPS}
P_L&&\equiv\left\{p\in \Lambda^*: \eps_L\rho^{1/3}\leq|p|\leq \eta^{-1}_L\rho^{1/3}\right\}\\\nonumber
P_H&&\equiv\left\{p\in \Lambda^*:\eps_H \leq   |p|\leq \eta_H^{-1}\right\}\, , 
\eeqa
where the parameters are chosen as follows 
\beq\label{defepseta}
\eps_L, \eta_L, \eps_H,\eta_H\equiv\rho^{\eta}\rmand \eta\equiv1/200
\eeq
\end{mydef}
We remark that the momenta between $P_L$ and $P_H$ are irrelevant to our construction and $\eta$ can be any positive number less than $1/200$. When $V=0$, most particles have momentum in $P_0\cup P_I\cup P_L$. When we turn on the interaction, pairs of these particles are annihilated and usually pairs of particles with momenta of order one will be created. 
\bigskip
\par Next, as in \cite{YY}, we define some notations for the states and subsets of the Fock space.  Using the  occupation number representation, we describe a state in Fock space with a function mapping from momentum space to integers. 
\begin{mydef}{Definition of $\widetilde M$,  $M$ and $N_\al$}\label{defwMMNal}
\par Let $P$ denote  $P_0\cup P_L\cup P_I\cup P_H$. We define  $\widetilde M$ as  the set of all functions $\al: P \rightarrow \N\cup 0$ such that
\beq
\sum_{k\in P}\al(k)=N
\eeq
For any $\al \in \widetilde M$,  denote by $|\al\rangle\in \mathcal H_{N,\Lambda}$  the unique state (in this case, an  $N$-particle  wave function) 
defined by the map $\alpha$ 
\[
|\al\rangle =  C \prod_{k\in P} (a^\dagger_k)^{\alpha(k)} | 0 \rangle\, , 
\]
where the positive constant $C$ is chosen so that $|\al\rangle $ is $L_2$-normalized. 
\par Moreover, we define $M$ as the  following subset of $\widetilde M$
\beq\label{defM}
M\equiv\{\al\in \widetilde M|{\rm supp}\,(\al)\subset P_0\cup P_I\cup P_L \rmand \al(k)\leq m_c\rmfor \forall k\in P_L\},
\eeq
where $m_c$ is defined as 
\beq\label{defmc}
m_c\equiv \rho^{-3\eta}=\rho^{-3/200}
\eeq
Clearly, we have 
\beq\label{number}
a^\dagger_k a_k|\al\rangle=\al(k)|\al\rangle,\,\,\, \forall k\in P
\eeq
\end{mydef}
The states corresponding to the functions in $M$, \eqref{defM}, have no particle with momentum of order one, and there is a restriction on the particle number.  But  when $V=0$, the total probability of finding the states corresponding to $M$ is almost equal to one. 
\par  Furthermore, as follows,  we can construct a trial state $\Gamma_0$, with $\al$'s in $M$, satisfying \eqref{resultPresult1} with $4\pi a$ replaced with $\int_{\R^3} \frac12 V dx$  in the r.h.s of \eqref{resultPresult1}.  We postpone the proof of the next lemma to the subsection \ref{proofresultPsial}. 
\begin{lem}\label{resultPsial}
For $\Lambda =[0, L]^3$, $L=\rho^{-\frac{41}{60}}$, $N=\rho^{-\frac{21}{20}}$ and $\lim_{\rho\to0}\rho^{2/3}\beta\in(0,\infty)$. There exists a state $\Gamma_0(\rho,\,\,\beta)$ having the form: ($g_\al(\rho,\beta)\in\R$)
\beq\label{resultPsialequ1}
\Gamma_0=\sum_{\al\in M} g_\al(\rho,\beta)|\al\rangle\langle\al|,\;\sum_{\al\in M} g_\al(\rho,\beta)=1,\;
\eeq
and satisfying 
\beq
\overline\lim_{\rho\to 0}\left(\frac1{|\Lambda|}\Tr\, H_N \Gamma_0 - \frac
1{|\Lambda|\beta} S(\Gamma_0)-f_{\,0}(\rho,\beta)\right)\rho^{-2}\label{prowgamma1}
\leq \frac12 V_0(2-[1-R[\beta]\,]^2_+)
\eeq
Furthermore, the coefficient function $g_\al$ satisfies
\beq\label{resultPsialequ2}
\lim_{\rho\to0}\sum_{\al\in M}N^{-2}N_\al g_\al=2-[1-R(\beta)\,]^2_+
\eeq
where we defined $N_\al\in \R$ ($\al\in M$) as 
\beq\label{defNal}
N_\al\equiv \al(0)\al(0)+\sum_{u,v\in P_L\cup P_0, u\neq \pm v}2\al(u)\al(v),\,\,\,\al\in M
\eeq
\end{lem}
We remark: actually $\Gamma_0$ is very close to $\Gamma_I$, the canonical Gibbs state of ideal Bose gases. The state $\Gamma_0(\rho,\,\,\beta)$ satisfies (\ref{prowgamma1}), but for all potentials $V\neq 0$, $V_0=\int V(x)dx^3$ is strictly larger than $8\pi a$. So we need to improve $\Gamma_0$. To do that, we need to  replace the $|\al\rangle$'s ($\al \in M$) with some non-product state $\Psi_\al$'s. The energy of $|\al\rangle$ is higher than what we really want, since in $|\al\rangle$ when two particles are close to each other their behavior does not look like $(1-w)$, which is the zero energy scattering solution of $V$. For this reason, we should construct $\Psi_\al$ as follows
\beqa\label{roughpsial}
&&\Psi_\al\sim C\prod_{i<j}(1-w)(x_i-x_j)|\al\rangle \\\nonumber
  &&\sim C\left(1-\sum_{i<j}w(x_i-x_j)+\sum w(x_i-x_j)w(x_k-x_l)\cdots\right)|\al\rangle\\\nonumber
  &&\sim C\left(\!\!\!1-\!\!\!\sum_{k}
  \frac{w_k}{|\Lambda|} \sum_{u,v}\cre_{u+k}\cre_{v-k}\ann_u\ann_v
  \!\!\!+\!\!\!\left(\sum_{k}\frac{w_k}{|\Lambda|}  \sum_{u,v}\cre_{u+k}\cre_{v-k}\ann_u\ann_v\right)^2\!\!\!\cdots\!\!\!\right)\!\!\!|\al\rangle
\eeqa
\par We give the rigorous definition in the next section. First, we noticed that the operator $\sum_{i<j}w(x_i-x_j)$ annihilates  two particles and creates  two new particles. In our temperature regime, usually the momenta of the annihilated particles are of order $\rho^{1/3}$ or zero, belong to $P_L\cup P_0$ and  momenta of the two created  particles are of order one, i.e., belong to $P_H$. With this fact,  we will construct $\Psi_\al$ as the linear combination of $\al$ and the states which can be obtained by keeping annihilating 2 particles with momenta in $P_L\cup P_0$ and creating 2 new particles with momentum of order one, i.e.,
\beqa\label{roughpsial2}
\Psi_\al\sim C\bigg(1-&&\sum_{k}
  \frac{w_k}{|\Lambda|} \sum_{u,v\in P_0\cup P_L}^{u+k, v-k\in P_H}\cre_{u+k}\cre_{v-k}\ann_u\ann_v\\\nonumber
 && +\left(\sum_{k}\frac{w_k}{|\Lambda|}  \sum_{u,v\in P_0\cup P_L}^{u+k, v-k\in P_H}\cre_{u+k}\cre_{v-k}\ann_u\ann_v\right)^2\cdots\bigg)|\al\rangle
\eeqa
\par  For simplicity, we divide the $P_H$ and $P_L$, which are subsets of momentum space, into small boxes. When the size of the boxes is small enough, the probability of finding two particles annihilated (created)  in same  box is extremely  low.  Therefore to construct $\Psi_\al$, we only use the states in which there is  at most one  particle  annihilated ($P_L$) or  created ($P_H$) in each small box. Now we define these boxes. 

\begin{mydef}\label{definition3}{Definitions of $B_H(u)$, $B_L(u)$ } 
\par Let $\kappa_L, \kappa_H>0$. Divide $P_L$ and $P_H$ \eqref{defPS} into small \textit{boxes} (could be non-rectangular box) s.t.\ the sides of the boxes are about $\rho^{\kappa_L}$ and $\rho^{\kappa_H}$. We denote the box containing $u$ by $B_H(u)$ when $u\in P_H$ ( $B_L(u)$ when $u\in P_L$).

\end{mydef}
\bigskip
Then we define the states which we will use to construct  $\Psi_\al$. 
\bigskip
\begin{mydef}{Definition of $\widetilde M_\al$
}\label{defwMal}
\par For any $\al\in M$, we define $\widetilde {M_\al}$ as the set of the $\be$'s  in $\widetilde M$ (Def. \ref{defwMMNal}) such that 
\begin{enumerate}
	\item If $k\in P_0$, then $\be(k)\leq \al(k)$. If $k\in P_I$, then $\be(k)=\al(k)$.
	\item There is \textbf{at most} one $k$ in each $B_L$ or $B_H$ satisfying $\be(k)\neq \al(k)$. 
	\item If $\be(k)\neq \al(k)$, then 
	\beqa\label{bealkpm1}
	&& \be(k)=\al(k)-1,\; \;\,\,\,\,\,\,\rmfor k\in P_L\\\nonumber
	&& \be(k)=\al(k)+1=1,\;\rmfor k\in P_H
	\eeqa
\end{enumerate}
\end{mydef}
As we explained, for each $\al\in M$, we  construct a normalized pure state $\Psi_\al$, which is a linear combination of $\be\in \widetilde M_\al$, i.e., 
\beq\label{psialform}
|\Psi_\al\rangle =\sum_{\be\in \widetilde{M_\al}}f_\al(\be)|\be\rangle,\;\,\,\,\,\,\, \sum_{\be\in \widetilde{M_\al}}|f_\al(\be)|^2=1
\eeq
\par 
To prove Theorem \ref{resultP}, i.e., to improve the $\Gamma_0$ in Lemma \ref{resultPsial}, we  choose the correct trial state $\Gamma$ of following form:
\beq\label{defGamma}
\Gamma=\sum_{\al\in M}g_{\al}|\Psi_\al\rangle\langle\Psi_\al|,
\eeq
where we choose $g_{\al}$ in \eqref{resultPsialequ1} and $\Psi_\al$ in \eqref{psialform}. 
\par With proper $\kappa_L$ and $\kappa_H$, $\Delta S$ the entropy difference between $\Gamma_0$ in \eqref{resultPsialequ1} and $\Gamma$ in \eqref{defGamma}  can be proved to be much less than $|\Lambda|\rho^2$. 

\begin{lem}\label{entropy} Let $\Lambda=\rho^{-41/20}$, $\kappa_L\leq 5/9$ and $\kappa_H\leq 2/9$.
Then    for any $\{\Psi_\al, \al\in M\}$ having the form \eqref{psialform} and any $g_\al>0$ such that $\sum_{\al\in M}g_{\al}=1$, we have 
   \beq\label{resultentropy}
 \overline \lim_{\rho\to0} \big[-S(\Gamma)-(-S(\Gamma_0))\big](\Lambda\rho^2)^{-1}= 0
   \eeq
  with $\Gamma$ defined in  \eqref{defGamma} and $\Gamma_0=\sum_{\al\in M} g_\al|\al\rangle\langle\al|$. 
\end{lem}
\par We postpone the proof of this lemma to subsection \ref{proofentropy}. The assumptions $\kappa_L\leq 5/9$ and $\kappa_H\leq 2/9$  imply
  \beq
 \rho^{1-4\eta-3\kappa_L}+\rho^{-4\eta-3\kappa_H} \ll N\rho^{1/3}.
  \eeq
\par In the next theorem, we  show that,  for each $\al\in M$, there exists a pure state $\Psi_\al$ of the form \eqref{psialform} such that,  comparing with $|\al\rangle$, the new pure state $|\Psi_\al\rangle$ lowers the total energy by about $(\frac12 V_0-4\pi a)N_\al\Lambda^{-1}$, where $N_\al$ is defined in \eqref{defNal}.  The construction of  the pure state  yielding the correct total energy  is the core of the proof of Theorem \ref{resultP}. 
\begin{thm}\label{mainlemma}
Let $1/2\geq \kappa_L\geq 4/9$ and $\kappa_H\geq 1/9$. For any $\al\in M$, there exists $\Psi_\al$ having the form \eqref{psialform} and  satisfying
 \beqa\nonumber
\langle \Psi_\al|H_N|\Psi_\al\rangle-\langle \al |H_N|\al\rangle+(\frac12 V_0-4\pi a)N_\al\Lambda^{-1}
\leq  \eps_\rho\rho^2\Lambda
\eeqa
where the $\eps_\rho$ is independent of $\al$ and $\lim_{\rho\to0}\eps_\rho=0$.
\end{thm}
\bigskip
\par Finally,  by choosing the proper size of the small boxes in $P_L$ and $P_H$, we can prove  Theorem \ref{resultP} with Theorem \ref{mainlemma},  Lemma \ref{mainlemma} and Lemma \ref{resultP}.
\bigskip
\par {\it{Proof of Theorem \ref{resultP}}
} 
\par Let $1/2\geq \kappa_L\geq 4/9$ and $2/9\geq\kappa_H\geq 1/9$. We choose   trial state $\Gamma$  \eqref{defGamma} with  $g_{\al}$ in Lemma \ref{resultPsial} \eqref{resultPsialequ1} and $\Psi_\al$'s in Theorem \ref{mainlemma}. Then
combine Theorem \ref{mainlemma},  Lemma \ref{mainlemma} and Lemma \ref{resultP}.
\qed
\bigskip
\par 
This paper is organized as follows: In Section \ref{deftps}, we  rigorously define $\Psi_\al$'s and the trial state $\Gamma$. 
In Section 5, we outline the lemmas needed to prove Theorem \ref{mainlemma}. 
In Section 6, we estimate the number of particles in the condensate and various momentum regimes. 
These estimates are the building blocks for all other estimates later on. The kinetic energy is estimated in Section 7 and the potential energy is estimated in Section 8-11. Finally in Section 12, we prove  Lemma \ref{relationDP}, \ref{resultPsial}, 
\ref{entropy}.

\section{Definition of the trial pure states $\Psi_\al$'s}\label{deftps}

In this section, we give a formal definition of the trial pure state  $\Psi_\al$'s for Theorem \ref{mainlemma}. 
For simplicity, we define  a special 'state' $ |{\bf 0}\rangle=0\in \mathcal H_{N,\Lambda}$. As in \cite{YY}, to construct $\Psi_\al$, we use the following   operators $A^{u,v}_{p,q}$:
 \beq
 A^{u,v}_{p,q}: \widetilde M\to\widetilde M\cup {\bf 0},\,\,\, u,v\in P_0\cup P_L,\;p,q\in P_H \rmand  u+v=p+q
 \eeq
With the notation $|{\bf 0}\rangle$, we have the following simple fomula for $A^{u,v}_{p,q}$, 
\beq
|A^{u,v}_{p,q}\be\rangle=Ca^\dagger_pa^\dagger_qa_ua_v|\be\rangle,\,\,\,\,\,\,\,\,\be \in \widetilde M
\eeq  
where $C$ is a positive normalization constant. We can see that, with the notation $\bf 0$, $A^{u,v}_{p,q}\be$ makes sense when the r.h.s is $0$. We note that here $\bf 0$ is introduced just for simplifying  the expression. 
\par The operator $A^{u,v}_{p,q}$ annihilates two particles with momenta in $P_L$ or $P_0$ and creates two particles with momenta in $P_H$.  We note: the total momentum is conserved. 
\par For simplicity, the pure trial state $\Psi_\al$ will be of the form  $\sum_{\be\in  M_\al}f_\al(\be)|\be\rangle$ where $f_\al$ is  supported in $M_\al\subset\widetilde M_\al$(def. \ref{defwMal}) which we now define. 
\par Note that there is no physical mean to construct $\Psi_\al$ on $M_\al$ and not $\widetilde M_\al$, but the properties of $M_\al$  simplify our proof. We can  define the coefficient function $f_\al$ on $M_\al$ with a very clear relation between $f_\al( \mathcal A^{u_1,u_2}_{k_1,k_2}\be )$ and $f_\al(\be )$, as in Lemma \ref{f}. But we can not do this on $\widetilde M_\al$.
\begin{mydef}\label{defmal}{Definition of nontrivial subset in $P_L$}
\par Let $A$ be a subset of $P_L$, it is called non-trivial when
\begin{enumerate}
	\item If $u_i\in A$ and  $u_i\neq u_j$($1\leq i\neq j\leq 2$), then $u_1+u_2\neq 0$
	\item If $u_i\in A$ and  $u_i\neq u_j$($1\leq i\neq j\leq 3$), then $u_1+u_2\neq u_3$
	\item If $u_i\in A$ and  $u_i\neq u_j$($1\leq i\neq j\leq 4$),  then $u_1+u_2\neq u_3+u_4.$
\end{enumerate}

\par Definition of $M_\al$: 
\par Recall $\widetilde M_\al$ in Def. \ref{defwMal}. For  $\al\in M$, we define the  subset $M_\al\subset \widetilde M_\al$ as the smallest set with the following properties.
\begin{enumerate}
		\item   For any  $\al$ and $\gamma\in \widetilde M_\al$, let $P_L(\gamma,\al)$ denote the following subset of $P_L$,
\beq\label{defPLgaal}
P_L(\gamma,\al)\equiv\{u\in P_L: \gamma(u)<\al(u)\}.
\eeq
Then  for any $\gamma\in M_\al$, $P_L(\gamma,\al)$ is non-trivial subset of $P_L$.
	\item $\al\in M_\al$
	\item If $\be\in M_\al $ and $\gamma=A^{0,0}_{p,-p}\be\in \widetilde M_\al$, then $\gamma\in M_\al$.
	\item If $\be\in M_\al $, $\gamma=A^{u,v}_{p,q}\be\in\widetilde M_\al$ and
	
\begin{enumerate}
	\item $P_L(\gamma,\al)$ is non-trivial
	\item $\be(-p)=\be(-q)=0$
\end{enumerate}
 then $\gamma\in M_\al$.
\end{enumerate}
\end{mydef}
Note: The set $M_\al$ is unique since the intersection of two such sets $M_{a\,1}$ and $M_{\al,2}$ satisfies all four conditions. 

\bigskip\par We collect a few obvious properties of the elements in $M_\al$ into the next lemma.
\begin{lem}\label{BpMal}
 By the definition of $M_\al$, any $\be\in M_\al$   has the following form:
\beq\label{Malform1}
\be=\prod_{i=1}^m\mathcal A^{u_{2i-1},u_{2i}}_{k_{2i-1},k_{2i}}\prod_{j=1}^n\mathcal A^{0,0}_{p_j,-p_j}\al, 
\eeq
where $u_i\in P_L\cup P_0$, $k_i\in P_H$ for $i=1,\cdots,2m$ and $p_j\in P_H$ for $j=1,\cdots,n$. And
\beq\label{Malform2}
p_i\neq \pm p_j, \;k_i\neq \pm k_j \rmfor i\neq j \rmand k_i\neq \pm p_j \rmfor \forall i,j
\eeq
\par On the other hand, if $\{u_i,(i=1,\cdots,2m) \}\cap P_L$ is a non-trivial subset of $P_L$, then any $\be\in \widetilde M_\al$ with form \eqref{Malform1} and \eqref{Malform2} belongs to $M_\al$. 
\par Furthermore, one can change the order of the $\cal A$'s in \eqref{Malform1}. With the fact that the subset of non-trivial subset of $P_L$ is still non-trivial, we can see, if $\be$ belongs to $ M_\al$ and has  the  form \eqref{Malform1} and \eqref{Malform2}, then we have
\beq\label{Malform3}
\prod_{i\in A}\mathcal A^{u_{2i-1},u_{2i}}_{k_{2i-1},k_{2i}}\prod_{j\in B}\mathcal A^{0,0}_{p_j,-p_j}\al\in M_\al
\eeq
Here $A$, $B$ are any subsets of $\{1,\cdots, m\}$ and $\{1,\cdots, n\}$
\end{lem}
\bigskip
Now, to define $\Psi_\al=\sum_{\be\in M_\al}f_\al(\be)|\be\rangle$, it only remains to define $f_\al$, which is supported on $M_\al$. As suggested in \eqref{roughpsial2}, for $u,v\in P_0\cup P_L$,  $p,q\in P_H$, and $u+v=p+q$, we have the following relation between $f_\al(\al)$ and $f_\al(\mathcal A^{u,v}_{p,q}\al)$
\beq\label{temp4.7.01}
f_\al(\mathcal A^{u,v}_{p,q}\al)\approx -(1-\delta_{u,v}/2)\left[w_{(p-u)}+w_{(v-p)}\right]|\Lambda|^{-1}\sqrt{\al(u)\al(v)}f_\al(\al)
\eeq
\par Furthermore, if $\be\in M_\al$ and $\sum_{k\in P_H}\be(k)$ is small (like $<5$), the approximation \eqref{roughpsial2} implies  that for most $u,v,p,q$,  
\beq\label{temp4.7}
f_\al(\mathcal A^{u,v}_{p,q}\be)\approx -(1-\delta_{u,v}/2)\left[w_{(p-u)}+w_{(v-p)}\right]|\Lambda|^{-1}\sqrt{\be(u)\be(v)}f_\al(\be)
\eeq
when $A^{u,v}_{p,q}\be\in M_\al$. Here we have used the fact that  when $\be\in M_\al$  and  $A^{u,v}_{p,q}\be\in M_\al$, $\be(p)=\be(q)=0$. 
\par We hope that for most $u,v\in P_0\cup P_L$,  $p,q\in P_H$, 
the approximation \eqref{temp4.7} would hold for most $\be\in M_\al$ such that $A^{u,v}_{p,q}\be\in M_\al$. Here ''most $\be$ have some property $A$'' means that the probability of finding $\be$ with this property in $\Psi_\al$ is almost one, i.e., 
\beq
\sum_{\be {\rm\,\,\, has\,\,\, property\,\,\,A}}|\langle \be|\Psi_\al\rangle|^2=\sum_{\be {\rm\,\,\, has\,\,\, property\,\,\,A}}|f_\al(\be)|^2\approx 1
\eeq
\par If the approximation \eqref{temp4.7} holds for some  $u,v\in P_0\cup P_L$,  $p,q\in P_H$, then we can easily obtain  
\beqa\label{temp4.9}
&&\langle \Psi_\al| \ann_u\ann_v\cre_p\cre_q|\Psi_\al\rangle\\\nonumber
&&\approx 
-(1-\delta_{u,v}/2)\left[w_{(p-u)}+w_{(v-p)}\right]|\Lambda|^{-1}\sum_{\be\in M_\al}\sqrt{\be(u)\be(v)}|f^2(\be)|
\eeqa
Using the definition of $M_\al$, we may guess that  that for most $\be\in M_\al$,
\beqa
&&\be(u)=\al(u),\,\,\,u,\in P_L\\\nonumber
&&\be(0)\sim\al(0)\,\,\,\,
\eeqa
Therefore 
\beq\label{temp4.10}
\langle \Psi_\al| \ann_u\ann_v\cre_p\cre_q|\Psi_\al\rangle\approx 
-(1-\delta_{u,v}/2)\left[w_{(p-u)}+w_{(v-p)}\right]|\Lambda|^{-1}\sqrt{\al(u)\al(v)}
\eeq
 This approximation \eqref{temp4.10} is very useful for calculating $\langle \Psi_\al| V|\Psi_\al\rangle$. 
\par Now we give the definition of $f_\al$ as follows.  In Lemma \ref{f} we  check that it has this property \eqref{temp4.7}.  
\begin{mydef}{\bf The Pure Trial State $\Psi_\al$}
\par Recall that the function $(1-w)$ is the zero energy scattering solution of the potential $V$, as in \eqref{zess}. Define the pure trial state  $\Psi_{\al}$ as
\beq
|\Psi_{\al}\rangle\equiv\sum_{\be\in  M_\al}f_\al(\be)|\be\rangle
\eeq
where the coefficient $f_\al(\be)$'s are given by 
 \beq\label{deffbe}
f_\al(\be)=C_\al 
\sqrt { \frac  {|\Lambda|^{\be(0)}} {\be(0)!} }\,\,\left(\prod_{k\in P_H}^{ \be(k)>0}\sqrt{-w_k}\right)\,\,\left(\prod_{ k\in P_H}^{\be(k)> \be(-k)}\sqrt{2}\right)\,\,\,\left(\prod_{u\in P_L(\be,\al)}\sqrt{\frac{\al(u)}{|\Lambda|}}\right)
\eeq
Here we follow the convention $\sqrt{x}=\sqrt{|x|}i$ for $x<0$.  For convenience, we define $f(\be)=0$ for $\be\notin M_\al$. 
The constant  $C_\al$ is  chosen so that $\Psi_\al$ is $L_2$ normalized, i.e., 
\[\left\langle \Psi_\al|\Psi_\al\right\rangle=1,\; i.e.,\;\sum_{\be\in M_\al}|f_\al(\be)|^2=1 \]  
\end{mydef}
In next Lemma, with the  $f_\al$ chosen  above, we show that \eqref{temp4.7} holds for most $u$, $v$, $p$, $q$, $\be$ such that $\be\in M_\al$  and  $A^{u,v}_{p,q}\be\in M_\al$.
\begin{lem} \label{f}
\begin{enumerate}
\item  If $k\in P_H$ and $\be\in M_\al, \mathcal A^{0,0}_{k,-k} \be  \in M_\al$, then 
	\beq \label{properf1}f_\al( \mathcal A^{0,0}_{k,-k}\be )=(-w_{k})\sqrt{\frac{\be(0)}{|\Lambda|}}\sqrt{\frac{\be(0)-1}{|\Lambda|}} f_\al(\be)
	\eeq
	
	\item  If $u_1, u_2\in P_L$, $u_2=\pm u_1$ or $u_2\in B_L(u_1)$, $k_1,k_2\in P_H$ and  $\be\in M_\al$, then $\gamma=\mathcal A^{u_1,u_2}_{k_1,k_2} \be  \notin M_\al$, i.e., $f_\al(\gamma)=0$.
	
	\item  If $u_1,u_2\in P_L\cup P_0$ and $u_2\neq \pm u_1$, $k_1,k_2\in P_H$, $\be\in M_\al$ and $\mathcal A^{u_1,u_2}_{k_1,k_2} \be  \in M_\al$,
	 then  when $\be(-p)=\be(-q)=0$, we have
	\beq \label{properf3}f_\al( \mathcal A^{u_1,u_2}_{k_1,k_2}\be )=2\sqrt{-w_{k_1}}\sqrt{-w_{k_2}}\sqrt{\frac{\be(u_1)}{|\Lambda|}}\sqrt{\frac{\be(u_2)}{|\Lambda|}} f_\al(\be)
	\eeq
	when $\be(-p)\neq 0$ or $\be(-q)\neq 0$, we have
	\beq \label{properf32}\left|f_\al( \mathcal A^{u_1,u_2}_{k_1,k_2}\be )\right|\leq \left|\sqrt{w_{k_1}}\sqrt{w_{k_2}}\sqrt{\frac{\be(u_1)}{|\Lambda|}}\sqrt{\frac{\be(u_2)}{|\Lambda|}} f_\al(\be)\right|
	\eeq
	\end{enumerate}
\end{lem}
 \par Again the result 2 in Lemma \ref{f} has no physical meaning, but it can simplify our proof. 
 \par 
 In next section, we can see that, for fixed $p\in P_H$ and  most $\be\in M_\al$,  $\be(-p)=0$. Hence the identity \eqref{properf1} or  \eqref{properf3} hold for most $\be\in M_\al$. Since $k_1, k_2$ are order one and $u_1,u_2\in P_0\cup P_L$, we have 
 \beq
 w_{k_1}\approx w_{k_2}\approx w_{k_1-u_1}\approx w_{k_1-u_2}=w_{u_2-k_1}
 \eeq
 which implies that $f_\al$ satisfies the property \eqref{temp4.7} in most case.

\section{Proof of Theorem \ref{mainlemma}}
\begin{proof}
Our goal is to prove 
\beqa\label{Maindesired}
\langle \Psi_\al|H_N|\Psi_\al\rangle-\langle \al |H_N|\al\rangle+(\frac12 V_0-4\pi a)N_\al\Lambda^{-1}
\leq  \eps_\rho\rho^2\Lambda
\eeqa 
\par First we decompose the Hamiltonian $H_N$ as in \cite{YY}. 
By the rule 1 of the definition of $\widetilde M_\al$, if $\be\in M_\al\subset \widetilde M_\al$ then $\be(k)$ is equal to $\al(k)$ for any $k\in P_I$.  Hence if $k_1\in P_I$, $ \be, \gamma\in M_\al$ and $\langle \beta|a^\dagger_{k_1}a^\dagger_{k_2}a_{k_3}a_{k_4}|\gamma\rangle\neq 0$, then one of $k_3$ and $k_4$ must be equal to $k_1$. 
\par On the other hand, since the particles with momenta in $P_H$ are created in pairs, the total number of the particles with momenta in $P_H$  is always even.
With these two results and momentum conservation , we can decompose the expectation value $\langle \Psi_\al|H_N|\Psi_\al\rangle$ as follows:
\beq\label{decoHPsial}
\langle H_N\rangle_{\Psi_\al} =\langle \sum_{i=1}^N-\Delta_i\rangle_{\Psi_\al}+\langle H_{abab}\rangle_{\Psi_\al}+\langle H_{\tilde L\tilde L}\rangle_{\Psi_\al}+\langle H_{\tilde LH}\rangle_{\Psi_\al}+\langle H_{HH}\rangle_{\Psi_\al},
\eeq
where 
\begin{enumerate}
	\item $H_{abab}$ is the part of interaction that annihilates two particles and creates the same two particles, i.e.,
	\beq\label{defHabab}
	H_{abab}= |2\Lambda|^{-1} \sum_{u}V_{0}a^\dagger_u a^\dagger_u a_u a_u+ |2\Lambda|^{-1} \sum_{u\neq v}(V_{u-v}+V_{0})a^\dagger_u a^\dagger_v a_u a_v
	\eeq
	\item $H_{\widetilde L\widetilde L}$ is the interaction between four particles with momenta in $P_{\widetilde L}$:
	\beq
	P_{\widetilde L}\equiv P_0\cup P_L
	\eeq
and 
	\beq
 H_{\widetilde L\widetilde L}=	|2\Lambda|^{-1}\sum_{u_i\in 	P_{\widetilde L}}V_{u_3-u_1}a^\dagger_{u_1} a^\dagger_{u_2} a_{u_3} a_{u_4},
	\eeq
	where $u_1\neq u_3$ or $u_4$. 
	\item $H_{\widetilde L H}$ is the part of interaction  that involves two particles with momenta in $P_{\widetilde L}$ and two particles with momenta in $P_{ H}$ i.e., 
	\beqa
	H_{\widetilde L H}=&&|2\Lambda|^{-1} \sum_{u_1,u_2\in P_{\widetilde L}, k_1,k_2\in P_H} V_{u_1-k_1}a^\dagger_{u_1} a^\dagger_{u_2} a_{k_1} a_{k_2}+H.C.\\\nonumber
	+&&|2\Lambda|^{-1} \sum_{u_1,u_2\in P_{\widetilde L}, k_1,k_2\in P_H} 2(V_{u_1-u_2}+V_{u_1-k_2})a^\dagger_{u_1} a^\dagger_{k_1} a_{u_2} a_{k_2},
	\eeqa
	where $u_1\neq u_2$ and $H.C.$  denotes the hermitian conjugate of the first term. 
	\item $H_{HH}$ is the part of   interaction between 4 particles with momenta in $P_{H}$, 
		\beq
 H_{HH}=	|2\Lambda|^{-1}\sum_{k_i\in 	P_{ H}}V_{k_3-k_1}a^\dagger_{k_1} a^\dagger_{k_2} a_{k_3} a_{k_4},
	\eeq
		where $k_1\neq k_3$ or $k_4$. 
\end{enumerate} 
With these definitions, since there is no high momentum particle in $|\al\rangle$ ($\al\in M$),  the total energy of $|\al\rangle$ is :
\beq\label{decomHal}
\langle\al| H_N|\al\rangle =\langle\al| \sum_{i=1}^N-\Delta_i|\al\rangle+\langle\al |H_{abab}|\al\rangle
\eeq
Recall the definition of $N_\al$ for $\al\in M$ in \eqref{defNal}. The estimates for the energies of these components   in \eqref{decoHPsial}  are stated as the following lemmas, which will be proved 
in later sections with different methods.

\begin{lem}\label{lemkinetic}The total kinetic energy is bounded  from above by
\beq\label{proofthem1-1}
\left\langle\sum_{i=1}^N-\Delta_i\right\rangle_{\Psi_\al}-\left\langle\sum_{i=1}^N-\Delta_i\right\rangle_{\al}-\|\nabla w\|^2_2 N_\al|\Lambda|^{-1}
\leq \eps_1\rho^2\Lambda,
\eeq
where $\eps_1$ is independent of $\al$ and $\lim_{\rho\to 0}\eps_1=0$.
\end{lem}

\begin{lem}\label{lemHabab}
The expectation value of $H_{abab}$ is bounded  above by,
\beq\label{proofthem1-2}
\left\langle H_{abab}\right\rangle_{\Psi_\al}-\left\langle H_{abab}\right\rangle_{\al}
\leq \rho^{11/4}\Lambda
\eeq
\end{lem}
\begin{lem}\label{lemHwLwL}
The expectation value of $H_{\wL\wL}$ is bounded  above by,
\beq\label{proofthem1-3}
\left\langle H_{\wL\wL}\right\rangle_{\Psi_\al}
\leq \rho^{11/4}\Lambda
\eeq
\end{lem}
\begin{lem}\label{lemHwLH}
The expectation value of $H_{\wL H}$ is bounded above by,
\beq\label{proofthem1-4}
\left\langle H_{\wL H}\right\rangle_{\Psi_\al}
+N_\al |\Lambda|^{-1}\|Vw\|_1
\leq \eps_2\rho^2\Lambda,
\eeq
where $\eps_2$ is independent of $\al$ and $\lim_{\rho\to 0}\eps_2=0$.
\end{lem}
\begin{lem}\label{lemHHH}
The expectation value of $H_{HH}$ is bounded above by,
\beq\label{proofthem1-5}
\left\langle H_{H H}\right\rangle_{\Psi_\al}
-N_\al|\Lambda|^{-1} \|\frac12Vw^2\|_1
\leq \eps_3\rho^2\Lambda,
\eeq
where $\eps_3$ is independent of $\al$ and $\lim_{\rho\to 0}\eps_3=0$.
\end{lem}
On the other hand, by definition of $w$  in \eqref{zess} and \eqref{defa}
, we have
\beq
 \|\nabla w\|_2^2-\|\frac12Vw\|_1+\|\frac12Vw^2\|_1=0, \;  
\frac12V_0-\|\frac12Vw\|_1=4\pi a
\eeq
Together with  \eqref{decomHal} and \eqref{proofthem1-1}-\eqref{proofthem1-5}, we arrive at the desired result \eqref{Maindesired}.
\end{proof}

\section{Estimates on the Numbers of Particles}
As in \cite {YY}, the first step to prove the Lemma \ref{lemkinetic} to Lemma \ref{lemHHH} is to estimate the particle number of $\Psi_\al$ in the condensate, $P_{L}, P_{I}$, and $P_{H}$. This is the main task of this section and 
we start with  the following  notations.

\begin{mydef}\label{1}
Suppose $u_i\in P=P_0\cup P_I\cup P_L\cup P_H$ for $i = 1, \ldots s$.  The expectation of the  product of  particle numbers with momenta  $u_1$, $\cdots$ $u_s$: 
	\beq\label{defQal}Q_\al  \left(u_1,u_2,\cdots,u_s\right)
\equiv\left\langle\prod_{i=1}^sa^\dagger_{u_i} a_{u_i}\right\rangle_{\Psi_\al}
=\sum_{\be\in M_\al} \prod_{i=1}^s{\be(u_i)|f_\al(\be)|^2}
\eeq

\end{mydef}
\begin{mydef}\label{MaluMBalu} The definition of $M_\al(u)$ and $M^B_\al(u)$
\par We denote by $M_\al(u)$  the set of $\be\in M_\al$'s satisfying $\be(u)=\al(u)$, i.e.
\beq\label{Malu}
M_\al(u)\equiv \{\be\in  M_\al: \be(u)=\al(u)\}
\eeq
Furthermore, with the definition of $B_L(u)$(when $u\in P_L$) and $B_H(u)$(when $u\in P_H$), we define $M^B_\al(u)\subset M_\al(u)$ as the intersection of $M_\al(v)$'s of all $v\in B_L(u)$ (when $u\in P_L$) or $B_H(u)$(when $u\in P_H$) , i.e.,  
\beq\label{MBalu}
M^B_\al(u)\equiv\cap_{v\in B_{L(H)}(u)}M_\al(v)
\eeq
We can see 
\beq
\beta\in M^B_\al(u)\;\Leftrightarrow \;\beta(v)=\al(v)\rmfor \forall v\in B_{L(H)}(u)
\eeq
\end{mydef}
The coefficient function $f_\al$ is supported on $M_\al\subset\widetilde M_\al$. Using  \eqref{bealkpm1}, if $\be\in M_\al$ and $u\in P_L$, either $\be(u)=\al(u)$, i.e., $\be\in M_\al(u)$ or $\be(u)=\al(u)-1$, i.e., $\be\notin M_\al(u)$. Therefore the average  number of the particles with momentum $u$,  for $u\in P_L$, can be written as follows
\beq
Q_\al(u)=\langle a^\dagger_{u} a_{u}\rangle_{\Psi_\al}=\al(u)-\sum_{\be\notin M_\al(u)}|f_\al(\be)|^2.
\eeq
 For any $k\in P_H$, we have 
 \beq
Q_\al(k)=\sum_{\be\notin M_\al(u)}|f_\al(\be)|^2.
\eeq 
The following theorem provides the main estimates on $Q_\al(u)$ and $Q_\al(k)$. 

\begin{lem}\label{roughbound} 
 For small enough $\rho$,  $Q_\al(u)$ and $Q_\al(k)$  can be estimated as follows ($u,u_1,u_2\in P_L$ and $k\in P_H$)
\beqa\label{boundspsiuroughPH}
 Q_\al(k)=\sum_{\be\notin M_\al(k)}|f_\al(\be)|^2&&\leq \const\rho^{2-4\eta},\!\!\rmfor k\in P_H\\
\label{boundspsiuroughPL}
0\leq \al(u)-Q_\al(u)=\sum_{\be\notin M_\al(u)}|f_\al(\be)|^2&&\leq \const\rho^{1-4\eta}, \rmfor u\in P_L
\eeqa
Furthermore, the probabilities of the combined cases are bounded as follows: ($u,u_1,u_2\in P_L$ and $k\in P_H$)
\beqa\label{boundspsiuvroughPL}
\sum_{\be\notin M_\al(u_1)\cup M_\al(u_2)}|f_\al(\be)|^2&&\leq \const \rho^{2-8\eta}\;\;{\rm when}\;u_1\neq u_2\\
\label{boundspsiukrough}
\sum_{\be\notin M_\al(u)\cup M_\al(k)}|f_\al(\be)|^2&&\leq\const \rho^{3-7\eta}|w_k|
\eeqa
\end{lem}

\begin{proof}{Proof of Lemma \ref{roughbound}}
\par First, we prove \eqref{boundspsiuroughPH} concerning $k\in P_H$. With Lemma \ref{BpMal}(\eqref{Malform1}-\eqref{Malform3}), when $\be(k)>0$,
there exist some $\gamma\in M_\al$ and $u,v\in P_L\cup P_0$, $p\in P_H$ such that  
\beq\label{auvkpgb}
\mathcal A^{u,v}_{k,p} \gamma=\be\;\rmand\;p=u+v-k
\eeq
With the properties of $f_\al$ in Lemma \ref{f}(\eqref{properf1}-\eqref{properf32}), $f_\al(\be)$ is  bounded as 
\beq\label{roughboundineq1}
|f_\al(\be)|^2\leq 4\gamma(u)\gamma(v)\Lambda^{-2}\left|w_{k}w_{p}\right||f_\al(\gamma)|^2.
\eeq
Then sum up $\be\notin M_\al(k)$, i.e.,  $\be(k)>0$, by summing up $u$, $v$ and $\gamma$, we obtain:
\beqa\nonumber
\sum_{\be\notin M_\al(k)}|f_\al(\be)|^2&\leq& 4\sum_{u,v\in P_L\cup P_0}\sum_{\gamma\in M_\al} \gamma(u)\gamma(v)\Lambda^{-2}\left|w_{k}w_{u+v-k}\right||f_\al(\gamma)|^2\\\label{temp6.13}
&\leq& 4\rho^2|w_{k}|\max_{p\in P_H}\{|w_p|\}
\eeqa
The upper bound of $|w_p|$ is derived  in \eqref{boundwp}: $\left|w_p\right|\leq 4\pi a |p|^{-2}$, therefore
\beq\label{Qalk}
Q_\al(k)=\sum_{\be\notin M_\al(k)}|f_\al(\be)|^2\leq \const \rho^{2-2\eta}|w_k| ,\; k\in P_H
\eeq
Using \eqref{boundwp} again, we obtain \eqref{boundspsiuroughPH}.
\par Then, we prove \eqref{boundspsiuroughPL} concerning $u\in P_L$. 
Similarly, with Lemma \ref{BpMal}, for any $\be\notin M_\al(u)$, i.e., $\be(u)=\al(u)-1$, there exist some $\gamma\in M_\al$ and $v\in P_L\cup P_0$, $p,k\in P_H$ such that  \eqref{auvkpgb} holds. This implies \eqref{roughboundineq1}. Using \eqref{boundwp} and $|k+p|=|u+v|\ll |k|$, we have
\beq\label{boundwpwk}
|w_pw_k|\leq \const|k|^{-4}, \,\,\,\,\,\,\,{\rm when} \;p,k\in P_H\;\rmand |p+k|\ll |k|
\eeq
Inserting \eqref{boundwpwk} and   the bounds $\gamma(u)\leq \al(u)\leq m_c=\rho^{-3\eta}$ into \eqref{roughboundineq1},  we obtain:
\beq
|f_\al(\be)|^2\leq \const\rho^{-3\eta}|k|^{-4}\gamma(v)\Lambda^{-2}|f_\al(\gamma)|^2
\eeq
Again, summing up $\be$(by summing up $\gamma$, $v$, $p$ and $k$), with $\sum_v\gamma(v)\leq N$, we obtain \eqref{boundspsiuroughPL} as follows
\beq\label{roughboundineq2}
\sum_{ \be\notin M_\al(u)}|f_\al(\be)|^2\leq \sum_{v\in P_L\cup P_0}^{ k\in P_H}\sum_{\gamma\in M_\al}\const\rho^{-3\eta}|k|^{-4}\gamma(v)\Lambda^{-2}|f_\al(\gamma)|^2\leq \rho^{1-4\eta}
\eeq 

\par Next, we prove \eqref{boundspsiuvroughPL} concerning $u_1,u_2\in P_L$. For any $\be\notin M_\al(u_1)\cup M_\al(u_2)$, i.e., 
\beq
\be(u_1)=\al(u_1)-1,\,\,\,\,\be(u_2)=\al(u_2)-1
\eeq
 using Lemma \ref{BpMal}, we can see that there are only two cases:
\begin{enumerate}
	\item there exist one $\gamma\in M_\al$, $p_1,p_2\in P_H$ and $\mathcal A^{u_1,u_2}_{p_1,p_2} \gamma=\be$ 
\item there exist one $\gamma\notin M_\al(u_2)$, $v\in P_L\cup P_0$, $v\neq u_2$, $p_1, p_2\in P_H$ and $\mathcal A^{u_1,v}_{p_1,p_2} \gamma=\be$ 
\end{enumerate}
As before, with the properties of $f_\al$ in Lemma \ref{f} 
, the bounds on $\al(u)$'s ($u\in P_L$) and \eqref{boundwpwk}, we have 
\beqa
\sum_{ \be\notin M_\al(u_1)\cup M_\al(u_2)}|f_\al(\be)|^2\leq 
&&\const\sum_{\gamma\in M_\al}\rho^{-7\eta}|\Lambda|^{-1}|f_\al(\gamma)|^2\\\nonumber
+&&\const\!\!\!\!\!\!\sum_{v\in P_L\cup P_0,\gamma\notin M_\al(u_2)}\rho^{-4\eta}\gamma(v)|\Lambda|^{-1}|f_\al(\gamma)|^2
\eeqa
Using $\sum_v\gamma(v)\leq N$ and \eqref{boundspsiuroughPL}, we obtain  \eqref{boundspsiuvroughPL}.
\par At last, we prove \eqref{boundspsiukrough} concerning $u\in P_L$ and $k\in P_k$. For any $\be\notin M_\al(u)\cup M_\al(k)$, Using Lemma \ref{BpMal}, we can see that there are only two cases:
\begin{enumerate}
	\item there exist $\gamma\in M_\al$, $v\in P_L\cup P_0$, $p\in P_H$ and $\mathcal A^{u,v}_{p,k} \gamma=\be$ 
\item there exist $\gamma\notin M_\al(u)$, $v_1,v_2\in P_L\cup P_0$, $p\in P_H$ and $\mathcal A^{v_1,v_2}_{p,k} \gamma=\be$ 
\end{enumerate}
Summing up $v$, $p$ or $v_1$, $v_2$, $p$, we obtain
\beqa\nonumber
\sum_{\be\notin M_\al(k)\cup M_\al(u)}|f_\al(\be)|^2
\leq&&\!\!\!\!\!\!\!\!\const\!\!\!\! \sum_{v\in P_L\cup P_0}\sum_{\gamma} \gamma(u)\gamma(v)\Lambda^{-2}|w_{k}w_{u+v-k}||f_\al(\gamma)|^2\\
+&& \!\!\!\! \sum_{\gamma\notin M_\al(u)}4\rho^2|w_{k}|\max_{p\in P_H}\{|w_p|\}|f_\al(\gamma)|^2
\eeqa
With the result in \eqref{boundwp}: $\left|w_p\right|\leq 4\pi a |p|^{-2}$ and   $\sum_v\gamma(v)\leq N$,  we have:
\beq
\sum_{\be\notin M_\al(k)\cup M_\al(u)}|f_\al(\be)|^2\leq  \const \gamma(u)\rho^{1-2\eta}\Lambda^{-1}|w_{k}|+\sum_{\gamma\notin M_\al(u)}4\rho^{2-2\eta}|w_{k}||f_\al(\gamma)|^2
\eeq
At last using \eqref{boundspsiuroughPL} and the fact $\gamma(u)\leq \al(u)\leq \rho^{-3\eta}$ and $\Lambda=\rho^{-41/20}$, we obtain the desired result  \eqref{boundspsiukrough}
\end{proof}
Moreover $Q_\al(k)$($k\in P_H$), has a more precise upper bound as follows.
\begin{lem}\label{boundPH}
For $k\in P_H$,  and $Q_\al(k)$ is bounded above by:
\beq\label{resultboundPH}
Q_\al(k)\leq N_\al \Lambda^{-2}|w_k|^2+\rho^{7/3-7\eta}
\eeq 
\end{lem}
\begin{proof} First using  Lemma \ref{BpMal}, we have that, for any $\be\notin M_\al(k)$, there are two cases:
\begin{enumerate}
	\item there exists $\gamma\in M_\al$, such that, $\mathcal A^{0,0}_{-k,k} \gamma=\be$ 
\item there exist $\gamma\in M_\al$, $u\neq \pm v\in P_L\cup P_0$,  $p\in P_H$, s.t., $\mathcal A^{u,v}_{p,k} \gamma=\be$.
\end{enumerate} 
Then with the identities and bound of  $f_\al$ in Lemma \ref{f} \eqref{properf1}, \eqref{properf3} and  \eqref{properf32} , $Q_\al(k)$ is bounded above by
\beq\label{temp6.17}
Q_\al(k)=\sum_{\be\notin M_\al(k)}|f_\al(\be)|^2\leq \al(0)^2\Lambda^{-2}w_k^2+\sum_{u,v\in P_L\cup P_0, u\neq \pm v}2\al(u)\al(v)\Lambda^{-2}|w_kw_{p}|
\eeq
where $p=u+v-k$. 
Since $w_p=w_{-p}$ and $|p+k|\leq 2(\rho^{1/3-\eta})$, with \eqref{derwp}, we have 
\beq
\left||w_k|-|w_p|\right|\leq \const \rho^{1/3-4\eta}
\eeq 
Inserting this into \eqref{temp6.17}, we obtain
\beq
Q_\al(k)\leq  N_\al\Lambda^{-2}w_k^2+\rho^{7/4-4\eta}|w_k|
\eeq 
Then using  $|w_k|\leq \const \rho^{-2\eta}$, we obtain the desired result \eqref{resultboundPH}.
\end{proof}
\bigskip

\par  
At last, with Lemma \ref{roughbound}, \ref{boundPH} and the definition of $M_\al$, one can easily obtain the following  inequalities on $f_\al$.
\begin{lem}\label{boundspsiBox}  Recall the definition of $M_\al^B(k)$ or $M_\al^B(u)$ in Def. \ref{MaluMBalu} \eqref{MBalu}, the upper bounds on $f_\al$ in  \eqref{boundspsiuroughPL} and \eqref{boundspsiuroughPH} imply:
\beq\label{boundspsiuroughPHB}
\sum_{\be\notin M_\al^B(k)}|f_\al(\be)|^2\leq  \rho^{2-4\eta}\Lambda\rho^{3\kappa_H}\leq \rho^{1/6}\rmfor k\in P_H
\eeq
and
\beq\label{boundspsiuroughPLB}
\sum_{\be\notin M_\al^B(u)}|f_\al(\be)|^2\leq  \rho^{1-4\eta}\Lambda\rho^{3\kappa_L}\leq \rho^{1/6}\rmfor u\in P_L
\eeq 
Recall $B_L$ and $B_H$ in Definition \ref{definition3}. Suppose $u_1$, $u_2\in P_L\cup P_0$, $k_1$, $k_2\in P_H$,  $u_1+u_2=k_1+k_2$, $u_1+u_2\neq 0$ and $u_1\notin B_L(u_2)$. Then using \eqref{boundspsiuroughPL}, \eqref{boundspsiuvroughPL} and the definition of $M_\al$, we have 
\beq\label{boundmaluv}
\sum_{\be\in M_\al,\,\,\, \mathcal A^{u_1,u_2}_{k_1,k_2}\be\notin M_\al}|f(\be)|^2
\leq\rho^{1/2}
\eeq
At last, with \eqref{boundspsiuroughPH} and the fact 
$$0\leq \al(0)-\be(0)\leq \sum_{k\in P_H}\be(k),$$ we have $Q_\al(0)$ and $Q_\al(0,0)$ bounded as follows 
\beq\label{boundq0}
\al(0)\geq Q_\al(0)\geq \al(0)-\rho^{5/6}N
\eeq
and
\beq\label{boundq00}
\big[\al(0)\big]^2\geq Q_\al(0,0)\geq \big[\al(0)\big]^2-N^2\rho^{5/6}
\eeq
\end{lem}

\section{Proof of Lemma \ref{lemkinetic}}
In this section, with the bounds on $Q_\al(u) (u\in P_L)$ and $Q_\al(k) (k\in P_H)$, we estimate the kinetic energy of $\Psi_\al$ by proving Lemma \ref{lemkinetic}. 
\begin{proof}
\par By the definition, 
\beq
\left\langle\sum_{i=1}^N-\Delta_i\right\rangle_{\Psi_\al}=\sum_{u\in P_L
\cup P_I\cup P_H} u^2 Q_\al(u)\rmand \left\langle\sum_{i=1}^N-\Delta_i\right\rangle_{\al}=\sum_{u\in P_L
\cup P_I} u^2 \al(u)
\eeq
 With the definition of $M_\al$ and $\widetilde M_\al$, we have  $Q_\al(u)\leq \al(u)$, for $u\in P_I\cup P_L$. Then the l.h.s of \eqref{proofthem1-1} bounded above by
 \beqa\label{temp7.2}
&&\left\langle\sum_{i=1}^N-\Delta_i\right\rangle_{\Psi_\al}-\left\langle\sum_{i=1}^N-\Delta_i\right\rangle_{\al}-\|\nabla w\|^2_2 N_\al|\Lambda|^{-1}\\\nonumber
\leq&& \sum_{k\in P_H} k^2Q_\al(k) -\|\nabla w\|^2_2 N_\al|\Lambda|^{-1}
\eeqa
With the upper bound on $Q_\al(k)$ in  \eqref{resultboundPH}, we have 
\beq
\eqref{temp7.2}\leq N_\al|\Lambda|^{-1}\left|\|\nabla w\|^2_2 -\sum_{k\in P_H}|\Lambda|^{-1}k^2|w_k|^2\right|+\rho^{13/6}\Lambda
\eeq
Together with  $\lim_{\rho\to0}\left|\|\nabla w\|^2_2 -\sum_{k\in P_H}|\Lambda|^{-1}k^2|w_k|^2\right|=0$, we complete the proof of Lemma \ref{lemkinetic}.
\end{proof}

\section{Proof of Lemma \ref{lemHabab}}
\begin{proof}
First we rewrite the expectation value of  $H_{abab}$ as
\beqa
&&\langle H_{abab}\rangle_{\Psi_\al}\\\nonumber
=&&|2\Lambda|^{-1}\!\!\!\sum_{\be \in M_\al}\!\!\!\left(V_0 \sum_{u}\left(\be(u)^2-\be(u)\right)+\sum_{u\neq v} (V_0+V_{u-v})\be(u)\be(v)\right)\!\!\!|f_\al(\be)|^2\\\nonumber
=&&|2\Lambda|^{-1}\sum_{\be \in M_\al}\left(V_0(N^2-N)+\sum_{u\neq v}V_{u-v}\be(u)\be(v)\right)|f_\al(\be)|^2
\eeqa
On the other hand, 
\beq
\langle H_{abab}\rangle_{\al}=|2\Lambda|^{-1}\left(V_0(N^2-N)+\sum_{u\neq v}V_{u-v}\al(u)\al(v)\right)
\eeq
By the assumptions, $V_v$ is positive when $|v|\ll 1$. For any $\be\in M_\al$,  $\be(u)\leq \al(u)$ for $u\in P_0\cup P_I\cup P_L$, therefore we have 
\beq
V_{u-v}\be(u)\be(v)\leq V_{u-v}\al(u)\al(v),\;{\rm when}\;u,v\in P_0\cup P_I\cup P_L
\eeq
Using this inequality and the fact $\al(k)=0$ for $k\in P_H$,  we have
\beqa\nonumber
&&\langle H_{abab}\rangle_{\Psi_\al}-\langle H_{abab}\rangle_{\al}\\\nonumber
&\leq& |2\Lambda|^{-1}\left(\sum_{u\notin P_H, v\in P_H } 2V_{u-v}Q_\al(u,v)+\sum_{u, v\in P_H } V_{u-v}Q_\al(u,v)\right)
\eeqa
For any $u\in P$,  $|V_u|$ is no more than $ |V_0|$, with \eqref{boundspsiuroughPH},  we obtain: 
\beqa
\langle H_{abab}\rangle_{\Psi_\al}-\langle H_{abab}\rangle_{\al}&\leq & V_0\rho\sum_{v\in P_H}Q_\al(v)\leq\rho^{11/4}\Lambda
\eeqa
\end{proof}
\bigskip
\section{Proof of Lemma \ref{lemHwLwL}}
As in \cite{YY}, to calculate $\langle a^\dagger_{u_1}a^\dagger_{u_2}a_{u_3} a_{u_4}\rangle_{\Psi_\al}$, we start  with  the following identity. 
\begin{lem}\label{id4a}
For any fixed momenta $u_{1,2,3,4}$ and $\be\in M_\al$,  define  $T(\be)$ to be the state 
\beq
|T(\be)\rangle\equiv Ca^\dagger_{u_1}a^\dagger_{u_2}a_{u_3} a_{u_4}|\be\rangle, 
\eeq 
where $C$ is the positive normalization constant when  $| T(\be)\rangle\not =0$. Then we have
\beq\label{resultid4a}
\langle a^\dagger_{u_1}a^\dagger_{u_2}a_{u_3} a_{u_4}\rangle_{\Psi_\al}=\sum_{\be\in M_\al}f_\al(\be)\overline{f_\al(T(\be))}
\sqrt{\langle\be|a^\dagger_{u_4}a^\dagger_{u_3}a_{u_2}a_{u_1} |a^\dagger_{u_1}a^\dagger_{u_2}a_{u_3} a_{u_4}|\be\rangle}
\eeq
\end{lem}
The map $T$ depends on $u_{1,2,3,4}$ and in principle it has to carry them as subscripts. We omit these subscripts since it will be clear
from the context  what they are. 

\begin{proof}
For any fixed $u_{1,2,3,4}$, by the definition of $\Psi_\al$, we have
\beq\label{proofid4a1}
\langle\Psi_\al|a^\dagger_{u_1}a^\dagger_{u_2}a_{u_3} a_{u_4}|\Psi_\al\rangle=\sum_{\gamma,\be\in M}f_\al(\be) {\overline{f_\al(\gamma)}}\langle\gamma|a^\dagger_{u_1}a^\dagger_{u_2}a_{u_3} a_{u_4}|\be\rangle
\eeq
By definition of $M_\al$, one can see 
\beq\label{proofid4a2}
\langle\gamma|a^\dagger_{u_1}a^\dagger_{u_2}a_{u_3} a_{u_4}|\be\rangle\neq 0\Rightarrow \gamma=T(\be)
\eeq
Since  $|T(\be)\rangle$ is normalized, the identity in Lemma \ref{id4a} is obvious. 
\end{proof}
\subsection{Proof of Lemma \ref{lemHwLwL}}
\begin{proof}
Using the fact $|V_u|\leq V_0$ for any $u\in \R^3$, we can see 
\beq
\left|\left\langle H_{\wL\wL}\right\rangle_{\Psi_\al}\right|\leq V_0|2\Lambda|^{-1}\sum_{u_i\in 	P_\wL,\, u_1\neq u_3,u_4}\left|\left\langle a^\dagger_{u_1} a^\dagger_{u_2} a_{u_3} a_{u_4}\right\rangle_{\Psi_\al}\right|,
\eeq
We are going to prove:
\beqa\label{00ll}
\sum_{u\in P_L}&&\left|\left\langle a^\dagger_0 a^\dagger_{0} a_{u} a_{-u}\right\rangle_{\Psi_\al}\right|=0 
\\\label{0lll}
\sum_{u_2, u_3,u_4\in P_L}&&\left|\left\langle a^\dagger_0 a^\dagger_{u_2} a_{u_3} a_{u_4}\right\rangle_{\Psi_\al}\right|\leq \Lambda^2\rho^{3-5\eta}
 \\\nonumber
\sum_{u_i\in P_L\rmand u_1\neq u_3,u_4}&&\left|\left\langle a^\dagger_{u_1} a^\dagger_{u_2} a_{u_3} a_{u_4}\right\rangle_{\Psi_\al}\right|\leq\Lambda^3\rho^{5-9\eta}  \\\label{llll}
\eeqa
\par First we note \eqref{00ll} is trivial. Because if $\be\in M_\al$, then $P_L(\be,\al)$ is non-trivial subset of $P_L$, which tells if $\be(u)<\al(u)$ then  $\be(-u)=\al(-u)$.
\par  Then we prove \eqref{0lll} concerning $u_{2,3,4}\in P_L$. 
By definition of $M_\al$, $$\langle \be|a^\dagger_0 a^\dagger_{u_2} a_{u_3} a_{u_4}|\gamma\rangle\neq 0$$
implies  $u_3\neq u_4$ and 
$\gamma\notin M_\al(u_2)$, i.e., $ \gamma(u_2)<\al(u_2)$. Furthermore, with the definition of $f_\al$ \eqref{deff}, we have 
\beq
f_\al(\be)=\sqrt{\frac{\al(u_3)\al(u_4)}{\be(0)\al(u_2)}}f_\al(\gamma)
\eeq
Combining with Lemma \ref{id4a}, we obtain 
\beq\label{0lllineq1}
\left|\left\langle a^\dagger_0 a^\dagger_{u_2} a_{u_3} a_{u_4}\right\rangle_{\Psi_\al}\right|\leq \al(u_3)\al(u_4)\sum_{\gamma\notin M_\al(u_2)}|f_\al(\gamma)^2|
\eeq
Using \eqref{boundspsiuroughPL} in Lemma \ref{roughbound}, we obtain 
\beq
\left|\left\langle a^\dagger_0 a^\dagger_{u_2} a_{u_3} a_{u_4}\right\rangle_{\Psi_\al}\right|\leq \const\al(u_3)\al(u_4)\rho^{1-4\eta},
\eeq
which implies \eqref{0lll}.
\par Next, we prove  \eqref{llll}. 
Similarly, we have 
\beq\label{llllineq1}
\left|\left\langle a^\dagger_{u_1} a^\dagger_{u_2} a_{u_3} a_{u_4}\right\rangle_{\Psi_\al}\right|\leq \al(u_3)\al(u_4)\sum_{\gamma\notin M_\al(u_1)\cup M_\al(u_2)}|f_\al(\gamma)^2|
\eeq
Again, using Lemma \ref{roughbound}, we obtain 
\beq
\left|\left\langle a^\dagger_{u_1} a^\dagger_{u_2} a_{u_3} a_{u_4}\right\rangle_{\Psi_\al}\right|\leq \const\al(u_3)\al(u_4)\rho^{2-8\eta},
\eeq
which implies \eqref{llll}. 
At last, combine \eqref{00ll}-\eqref{llll} and we obtain 
\beq
|\left\langle H_{\wL\wL}\right\rangle_{\Psi_\al}|\leq \rho^{11/4}\Lambda
\eeq
\end{proof}
\section{Proof of Lemma \ref{lemHwLH}}
We start the proof  with estimating $\langle a^\dagger_{u_1} a^\dagger_{u_2} a_{k_1} a_{k_2}\rangle_{\Psi_\al}$ in the special case: $u_1=\pm u_2\in P_L$. By the definition of $M_\al$, if $\be\in M_\al$, $u\in P_L$ and $\be(u)<\al(u)$, then $\be(u)=\al(u)-1$ and $\be(-u)=\al(-u)$. Since $f_\al$ is supported on $M_\al$, we have:
\beq\label{temp10.1}
\langle a^\dagger_{u_1} a^\dagger_{u_2} a_{k_1} a_{k_2}\rangle_{\Psi_\al}=0,\; \rmfor \forall k_1,\;k_2\in P_H, \;u_1=\pm u_2\in P_L
\eeq
\par For the other cases, we leave the bounds in the following lemma. As explained before, with the $f_\al$ we chose, the approximation \eqref{temp4.10} should hold for most $u,v\in P_L\cup P_0$, $p,q\in P_H$. In the proof of Lemma \ref{lemHwLoc}, one can see that the approximation \eqref{temp4.10} implies the main results \eqref{lemHwLHuk} and \eqref{lemHwLHuukk}.
\begin{lem}\label{lemHwLoc} Recall $P_{\widetilde L}=P_0\cup P_L$. For $u,u_1,u_2\in P_{\widetilde L}$ and $k,k_1,k_2\in P_H$, we have
	\beq\label{lemHwLHuk}
	 \left|\sum
	  V_{u-k}\langle a^\dagger_{u} a^\dagger_{-u} a_{k} a_{-k}\rangle_{\Psi_\al}+\al(0)^2 \|Vw\|_1\right|\leq \eps_4 N^2
	\eeq
	\beqa\label{lemHwLHuukk}
	\left|\sum_{
	u_1\neq \pm u_2}
	V_{u_1-k_1}\langle a^\dagger_{u_1} a^\dagger_{u_2} a_{k_1} a_{k_2}\rangle_{\Psi_\al}\!\!\!+\!\!\!\!\!\!\sum_{
	u_1\neq \pm u_2}\!\!2\al(u_1)\al(u_2) \|Vw\|_1\right|\leq \eps_5 N^2
\eeqa
and \beqa	\label{lemHwLHukuk}
	\sum_{
	u_1\neq u_2} \left|\langle a^\dagger_{u_1} a^\dagger_{k_1} a_{u_2} a_{k_2}\rangle_{\Psi_\al}\right|\leq \eps_6 N^2
	\eeqa
	where we omitted $u,u_1,u_2\in P_{\widetilde L}$, $k,k_1,k_2\in P_H$ and momentum conservation equality in $\sum$. The small numbers $\eps_4, \eps_5, \eps_6$ are independent of $\al$ and $\lim_{\rho\to 0}\eps_i=0$ for $i=4,5,6$.
\end{lem}
\begin{proof}{Proof of Lemma \ref{lemHwLH}}
\par Combine the bounds in \eqref{temp10.1}, \eqref{lemHwLHuk}, \eqref{lemHwLHuukk} and \eqref{lemHwLHukuk}.
\end{proof}
\subsection{Proof of Lemma \ref{lemHwLoc}}
\begin{proof} 
First we prove \eqref{lemHwLHuk} concerning $u\in P_\wL$ and $k\in P_H$. By \eqref{temp10.1}, if $\langle a^\dagger_{u} a^\dagger_{-u} a_{k} a_{-k}\rangle_{\Psi_\al}\neq 0$, then $u$ must be zero. The property of $f_\al$ in Lemma \ref{f} \eqref{properf1} implies
\beq
\langle \be|a^\dagger_{0} a^\dagger_{-0} a_{k} a_{-k}|\gamma\rangle\neq 0 \Rightarrow
\frac{f_\al(\gamma)}{f_\al(\be)}=-\frac{w_k}{|\Lambda|}\sqrt{\gamma(0)^2-\gamma(0)}
\eeq
Together with Lemma \ref{id4a}, we have 
\beq
\langle a^\dagger_{0} a^\dagger_{0} a_{k} a_{-k}\rangle_{\Psi_\al}=
-w_k\sum_{\be:\,\be\in M_\al,\,\mathcal A^{0,0}_{k,-k}\be\in M_\al}\left(\be(0)^2-\be(0)\right)\Lambda^{-1}|f_\al(\be)|^2,
\eeq
Recall the definitions of $M_\al^B$'s in Def. \ref{defwMal}. One can see if $\be(0)>1$, then $\be\in M_\al$ and $\mathcal A^{0,0}_{k,-k}\be\in M_\al$ is equivalent to $\be\in M^B_\al(k)\cap M^B_\al(-k)$. Therefore, we have the following identity,
\beq\label{sotemp10.7}
\langle a^\dagger_{0} a^\dagger_{0} a_{k} a_{-k}\rangle_{\Psi_\al}=
-w_k\sum_{\be\in M^B_\al(k)\cap M^B_\al(-k)}\left(\be(0)^2-\be(0)\right)\Lambda^{-1}|f_\al(\be)|^2,
\eeq
Using the bound on $\sum_{\be\notin M_\al^B(k)}|f_\al(\be)|^2$ \eqref{boundspsiuroughPHB} and the bounds on $Q_\al(0)$, $Q_\al(0,0)$ in \eqref{boundq0} and \eqref{boundq00}. We obtain that 
\beq\label{sotemp10.8}
\left|\sum_{\be\in M^B_\al(k)\cap M^B_\al(-k)}\left(\be(0)^2-\be(0)\right)|f_\al(\be)|^2-\al(0)^2\right|\leq O(\rho^{1/6}N^2)
\eeq
Insert \eqref{sotemp10.8} into \eqref{sotemp10.7}. Then summing up $k\in P_H$, with $u=0$, we obtain
\beqa
&&\left|\sum
 V_{u-k}\langle a^\dagger_{u} a^\dagger_{-u} a_{k} a_{-k}\rangle_{\Psi_\al}+\al(0)^2 \|Vw\|_1\right|\\\nonumber
 \leq &&\al(0)^2\left|\sum_{k\in P_H}-V_kw_k\Lambda^{-1}+\|Vw\|_1\right|+O(\rho^{1/6-3\eta}N^2) 
\eeqa
Combining  with the fact $\lim_{\rho\to 0}\left|\sum_{k\in P_H}-V_kw_k\Lambda^{-1}+\|Vw\|_1\right|=0$, we obtain the desired result \eqref{lemHwLHuk}.
\par Next, we prove \eqref{lemHwLHuukk} concerning $u_1,u_2\in P_\wL$, $u_1\neq \pm u_2$ and $k_1,k_2\in P_H$. Using the result 2 in Lemma \ref{f}, one can see
\beq
\langle a^\dagger_{u_1} a^\dagger_{u_2} a_{k_1} a_{k_2}\rangle_{\Psi_\al}=0 {\rm \;when\; }u_2\in B_L(u_1)
\eeq
Then from now on, we assume $u_2\notin B_L(u_1)$. The property of $f_\al$ in  Lemma \ref{f} implies, when  $\langle \be|a^\dagger_{u_1} a^\dagger_{u_2} a_{k_1} a_{k_2}|\gamma\rangle\neq 0 $ and $\be,\gamma\in M_\al$, 
\beq
f(\gamma)=C_\beta\sqrt{-w_{k_1}}\sqrt{-w_{k_2}}\sqrt{\be(u_1)\be(u_2)} f(\be)
\eeq
Here $C_\beta$ depends on $\be$ and $|C_\beta|\leq 2$. Especially, when $\be\in M_\al(-k_1)\cap M_\al(-k_2)$, $C_\be=2$. Again  with Lemma \ref{id4a}, for fixed $u_1$, $u_2\notin B_L(u_1)$, $k_1$ and  $k_2$, we have
\beq\label{whtemp1012}
\langle a^\dagger_{u_1} a^\dagger_{u_2} a_{k_1} a_{k_2}\rangle_{\Psi_\al}=\sqrt{-w_{k_1}}\sqrt{-w_{k_2}}\sum_{\be\in M_\al, \;\mathcal A^{u_1,u_2}_{k_1,k_2}\be\in M_\al}C_\be\be(u_1)\be(u_2)|f(\be)|^2,
\eeq
First,  using the facts $\left|k_1+k_2\right|\leq 2\rho^{1/3}\eta_L^{-1}$ and the bound on $dw_p/dp$ \eqref{derwp}, we obtain $|w_{k_1}-w_{k_2}|\leq \rho^{1/4} $, 
therefore
\beq\label{sqrtk12k1}
\left|\left(\sqrt{-w_{k_1}}\sqrt{-w_{k_2}}\right)+w_{k_1}\right|\leq \rho^{1/4}
\eeq
Insert \eqref{sqrtk12k1} into \eqref{whtemp1012}, we have 
\beq\label{temp10.14}
\langle a^\dagger_{u_1} a^\dagger_{u_2} a_{k_1} a_{k_2}\rangle_{\Psi_\al}=(-w_{k_1}+O(\rho^{1/4}))\sum_{\be\in M_\al, \;\mathcal A^{u_1,u_2}_{k_1,k_2}\be\in M_\al}\!\!\!\!\!C_\be\be(u_1)\be(u_2)|f_\al(\be)|^2.
\eeq
Now we bound $$\sum_{\be\in M_\al, \;\mathcal A^{u_1,u_2}_{k_1,k_2}\be\in M_\al}\!\!\!\!\!C_\be\be(u_1)\be(u_2)|f_\al(\be)|^2.$$ In the case $\be\notin M_\al(-k_1)\cap M_\al(-k_2)$, using the result in \eqref{boundspsiuroughPH} and $|C_\beta|\leq 2$, we have  
\beq\label{temp10.15}
\left|\sum_{\be\notin M_\al(k_1)\cap M_\al(k_2)}C_\be\be(u_1)\be(u_2)|f_\al(\be)|^2\right|\leq \rho\,\al(u_1)\al(u_2)
\eeq
In the case $\be\in M_\al(-k_1)\cap M_\al(-k_2)$, we have $C_\be=2$. 
Using the results in Lemma \ref{roughbound} and Lemma \ref{boundspsiBox}(\eqref{boundspsiuroughPH}, \eqref{boundspsiuroughPHB}, \eqref{boundspsiuroughPLB}, \eqref{boundmaluv} and $\al(u)\leq m_c=\rho^{-3\eta}$ for $u\in P_L$,  we obtain that if $u_1, u_2\in P_L$
\beq\label{temp9.13}
\left|\sum^{\be\in M_\al(-k_1)\cap M_\al(-k_2)}_{A^{u_1,u_2}_{k_1,k_2}\be\in M_\al}\be(u_1)\be(u_2)|f_\al(\be)|^2-\al(u_1)\al(u_2)\right|\leq O(\rho^{1/6-6\eta})
\eeq
and if $u_1=0, u_2\in P_L$, we have
\beq\label{temp9.14}
\left|\sum^{\be\in M_\al(-k_1)\cap M_\al(-k_2)}_{A^{u_1,u_2}_{k_1,k_2}\be\in M_\al}\be(u_1)\be(u_2)|f_\al(\be)|^2-\al(u_1)\al(u_2)\right|\leq O(\rho^{1/6-3\eta}N)
\eeq
Inserting \eqref{temp10.15}, \eqref{temp9.13} and \eqref{temp9.14} into \eqref{temp10.14}, with the fact $|w_p|\leq 4\pi a |p|^{-2}$,  we obtain that  
for $u_1, u_2\in P_L$:
\beq
\left|\langle a^\dagger_{u_1} a^\dagger_{u_2} a_{k_1} a_{k_2}\rangle_{\Psi_\al}+2w_{k_1}\al(u_1)\al(u_2)\right|\leq O(\rho^{1/6-8\eta})
\eeq
and for $u_1=0, u_2\in P_L$, 
\beq
\left|\langle a^\dagger_{u_1} a^\dagger_{u_2} a_{k_1} a_{k_2}\rangle_{\Psi_\al}+2w_{k_1}\al(u_1)\al(u_2)\right|\leq O(\rho^{1/6-5\eta}N).
\eeq
Furthermore,  the  smoothness and symmetry of $V$ implies $$|V_{u_1-k_1}-V_{k_1}|\leq\rho^{1/4}.$$ 
Then summing up $u_1,u_2: u_2\notin B_L(u_1)$ and  $k_1$, $k_2$, we obtain
\beqa\nonumber
&&\left|\sum_{u_1\neq\pm u_2
}V_{u_1-k_1} \langle a^\dagger_{u_1} a^\dagger_{u_2} a_{k_1} a_{k_2}\rangle_{\Psi_\al}
+2\sum_{
u_1\neq \pm u_2}\al(u_1)\al(u_2) \|Vw\|_1\right|\\\nonumber
 \leq &&2\sum_{
 u_1\neq \pm u_2}\left(\al(u_1)\al(u_2)\left|\sum_{
 }|V_{k_1}w_{k_1}|\Lambda^{-1}- \|Vw\|_1\right|\right)+O(\rho^{1/6-17\eta}N^2)\\
 +&& \sum_{\{u_1,u_2:\, u_2\in B_L(u_1)\}}2\al(u_1)\al(u_2)\|Vw\|_1
\eeqa
One can see the first line of the r.h.s  is less than $\eps_5N^2/2$. Here $\eps_5$ is independent of $\al$ and $\lim_{\rho\to0}\eps_5=0$. With the bound $\al(u)\leq m_c$ for $u\in P_L$, we can obtain that the second line of the right side is also $o(N^2)$. Therefore we arrive at the desired result \eqref{lemHwLHuukk}.
\par At last, we prove \eqref{lemHwLHukuk} concerning $u_{1,2}\in P_L$, $u_1\neq u_2$ and $k_{1,2}\in P_H$. The definitions of $M_\al$ and $f_\al$ imply that, when $\langle\be| a^\dagger_{u_1} a^\dagger_{k_1} a_{u_2} a_{k_2}|\gamma\rangle\neq 0$ and $\be,\gamma\in M_\al$,
 $$\gamma\notin M_\al(u_1)\cup M_\al(k_2)\,\,\,\be\notin M_\al(u_2)\cup M_\al(k_1)$$
  and 
\beq
|f_\al(\gamma)|\leq \const \left|\sqrt{\frac{\al(u_1)}{\al(u_2)}}\sqrt{\frac{w_{k_2}}{w_{k_1}}}\right||f_\al(\be)|
\eeq
This implies
\beq
\left|f_\al(\be)f_\al(\gamma)\langle\be| a^\dagger_{u_1} a^\dagger_{k_1} a_{u_2} a_{k_2}|\gamma\rangle\right|
\leq \const\al(u_1)\left|\sqrt{\frac{w_{k_2}}{w_{k_1}}}\right||f_\al(\be)|^2
\eeq
Summing up $\be\notin M_\al(u_2)\cup M_\al(k_1)$, with the upper  bound on $\sum_{\be}|f_\al(\be)|^2$ \eqref{boundspsiukrough}, we have
\beqa\nonumber
\left|\langle a^\dagger_{u_1} a^\dagger_{k_1} a_{u_2} a_{k_2}\rangle_{\Psi_\al}\right|\leq &&\const \al(u_1)\left|\sqrt{\frac{w_{k_2}}{w_{k_1}}}\right|\sum_{\be\notin M_\al(u_2)\cup  M_\al(k_1)}|f_\al(\be)|^2\\
\leq &&\al(u_1)|\sqrt{w_{k_1}w_{k_2}}|\rho^{3-8\eta}
\eeqa
At last, using $|w_p|\leq 4\pi a |p|^{-2}$ and $|k_1|\sim|k_2|$, we have 
\beqa\nonumber
	\sum_{u_1\neq u_2
	} \left|\langle a^\dagger_{u_1} a^\dagger_{k_1} a_{u_2} a_{k_2}\rangle_{\Psi_\al}\right|\leq&&\sum_{u_1,u_2,k_1,k_2} \al(u_1)\rho^{3-10\eta}\\
	\leq&&\Lambda^3\rho^{5-13\eta}=o(\Lambda^2\rho^{5/2})
	\eeqa
\end{proof}

\section{Proof of Lemma \ref{lemHHH}}
In this section, we will prove Lemma \ref{lemHHH} involving interaction energy between particles with momenta in $P_{H}$. 
We will show that the only contribution to the accuracy we need comes from four high momentum particles, 
to  be computed in Lemma \ref{HHHmain} \eqref{resultHHHmain}. We start with separating $\langle H_{HH}\rangle_{\Psi_\al}$ into the main terms and the error terms. 
\par Define $M_\al(k_1,k_2,k_3,k_4,u_1,u_2)\subset M_\al\otimes M_
\al$ as the set of $(\be,\gamma)$'s where $\be$ and $\gamma$ can be created from the same $\tilde \al\in M_\al$ as follows, 
\beqa\label{defMalkkuu}
&&M_\al(k_1,k_2,k_3,k_4,u_1,u_2)\\\nonumber
\equiv &&\{(\be, \gamma)\in M_\al\otimes M_\al: \exists \,\tilde \al\in M_\al {\,\,\,\rm s.t.\,\,\,} \mathcal A^{u_1,u_2}_{k_1,k_2}\tilde\al=\be \rmand A^{u_1,u_2}_{k_3,k_4}\tilde\al=\gamma\},
\eeqa
where $k_1,k_2, k_3,k_4\in P_H$ and $u_1,u_2\in P_\wL$. 
We define $A_{u_1,u_2, k_1,k_2, k_3,k_4}$ as
\beq\label{Au1u2k1k2k3k4}
A_{u_1,u_2, k_1,k_2, k_3,k_4}\equiv \sum_{(\be,\gamma)\in M_\al(k_1,k_2,k_3,k_4,u_1,u_2)}\overline{f_\al(\be)}f_\al(\gamma)\left\langle \be |a^\dagger_{k_1}a^\dagger_{k_2}a_{k_3}a_{k_4}|\gamma\right\rangle
\eeq
We note: 
\beq
 \left\langle  a^\dagger_{k_1}a^\dagger_{k_2}a_{k_3}a_{k_4}\right\rangle_{\Psi_\al}=\sum_{\be,\gamma\in M_\al}\overline{f_\al(\be)}f_\al(\gamma)\left\langle \be |a^\dagger_{k_1}a^\dagger_{k_2}a_{k_3}a_{k_4}|\gamma\right\rangle
\eeq
With \eqref{Au1u2k1k2k3k4}, we can separate the expectation value of $H_{HH}$ into two parts, main term (Lemma \ref{HHHmain}) and error term (Lemma \ref{HHHerror}). 
\begin{lem}\label{HHHmain}
Summing up $k_1, k_2, k_3, k_4\in P_H$,  $k_i\neq k_j$ for $i\neq j$, $u_1, u_2\in P_\wL$, we have
\beq\label{resultHHHmain}
\left|\sum_{u_i,k_i} V_{k_1-k_3}\Lambda^{-1}A_{u_1,u_2, k_1,k_2,k_3,k_4}-N_\al|\Lambda|^{-1} \|Vw^2\|_1\right|
\leq \frac{\eps_3}{2}\rho^2\Lambda,
\eeq
where $\eps_3$ is independent of $\al$ and $\lim_{\rho\to 0}\eps_3=0$.
\end{lem}
\begin{lem}\label{HHHerror}
 Let $M_\al(k_1,k_2,k_3,k_4)$ be the union of $M_\al(k_1,k_2,k_3,k_4,u_1,u_2)$, i.e., 
 \beq\label{defmalkkkk}
 M_\al(k_1,k_2,k_3,k_4)\equiv \cup_{u_1,u_2\in P_\wL}M_\al(k_1,k_2,k_3,k_4,u_1,u_2).
 \eeq
 Then we have
\beq\label{HHHerror1}
\sum_{k_i\in P_H}\sum_{(\be,\gamma)\notin M_\al(k_1,k_2,k_3,k_4)}V_0\Lambda^{-1}\left|\overline{f_\al(\be)}f_\al(\gamma)\left\langle \be |a^\dagger_{k_1}a^\dagger_{k_2}a_{k_3}a_{k_4}|\gamma\right\rangle\right|\leq\frac{\eps_3}{2}\rho^2\Lambda
\eeq
Here $k_i\neq k_j$ for $i\neq j$ and $\eps_3$ is independent of $\al$, $\lim_{\rho\to 0}\eps_3=0$.
\end{lem}
\subsection{Proof of Lemma \ref{lemHHH}}
\begin{proof}
\par  Definition of $M_\al$ implies that when $k\in P_H$ and $\be\in M_\al$, 
$$\be(k)\in \{0, 1\}$$  
Then the expectation value of $a^\dagger_{k_1}a^\dagger_{k_2}a_{k_3}a_{k_4}$ must be zero when $k_1=k_2$ or $k_3=k_4$. Together with the definition of $H_{HH}$,  we can rewrite $\langle H_{HH}\rangle_{\Psi_\al}$ as 
\beq
\langle H_{HH}\rangle_{\Psi_\al}=\sum_{k_i\in P_H}^{k_i\neq k_j}\sum_{\be,\gamma\in M_\al}\frac12V_{k_1-k_3}\Lambda^{-1}\overline{f_\al(\be)}f_\al(\gamma)\left\langle \be |a^\dagger_{k_1}a^\dagger_{k_2}a_{k_3}a_{k_4}|\gamma\right\rangle
\eeq
On the other hand,  if $\be, \gamma\in M_\al$ and  
$\langle \be|a^\dagger_{k_1}a^\dagger_{k_2}a_{k_3}a_{k_4}|\gamma\rangle\neq 0$ for some $k_{1,2,3,4}\in P_H$, then by the fact $P_L(\be,\al)=P_L(\gamma,\al)$ is non-trivial subset of $P_L$(Def. \ref{defmal}),  there exists \textbf{at most} one pair of $\{u_1,u_2\}$ such that
\beq\label{temp11.6}
(\be,\gamma)\in M_\al(k_1,k_2,k_3,k_4,u_1,u_2)
\eeq
Therefore combining \eqref{resultHHHmain} and \eqref{HHHerror1}, with $|V_{k_1-k_3}|\leq V_0$, we obtain the desired result \eqref{proofthem1-5}. 

\end{proof}
\subsection{Proof of Lemma \ref{HHHmain}}
\begin{proof}
We start with bounding  $A_{u_1,u_2, k_1,k_2, k_3,k_4}$. 
\begin{lem}\label{boundA}
When $u_1, u_2\in P_L$ and $u_1=\pm u_2$ or $u_2\in B_L(u_1)$, for any $k_i\in P_H$, we have
\beq\label{boundauukkkk0}
A_{u_1,u_2, k_1,k_2, k_3,k_4}=0
\eeq
In other cases,  $A_{u_1,u_2, k_1,k_2, k_3,k_4}$ is bounded by (Recall $P_0=\{0\}$)
\beqa\label{boundauukkkk}
&&\left|A_{u_1,u_2, k_1,k_2,k_3,k_4}
-\al(u_1)\al(u_2)F_a(u_1,u_2)^2w_{k_1}w_{k_3}\Lambda^{-2}\right|
\\\nonumber
\leq && \rho^{1/8}\Lambda^{-2}\times\left\{\begin{array}{ll}
 \al(u_1)\al(u_2),&u_1,u_2\in P_L\\  N\al(u_2),& u_1\in P_0,u_2\in P_L\\  N^2,& u_1=u_2\in P_0,
\end{array}\right.
\eeqa
where $F_a(u_1,u_2)=1$ when $u_1=u_2=0$, otherwise $F_a(u_1,u_2)=2$. 
\end{lem}
\begin{proof}{Proof of Lemma \ref{boundA}}
\par First we prove \eqref{boundauukkkk0}.  One can see that it follows  the definition of $A_{u_1,u_2, k_1,k_2, k_3,k_4}$ and the  result 2 in Lemma \ref{f}.
\par Then we prove \eqref{boundauukkkk} when $u_1,u_2\in P_L$. When \eqref{temp11.6} holds, by the definition of $M_\al$ $(k_1,k_2,k_3,k_4,u_1,u_2)$ in \eqref{defMalkkuu},   there exists $\tilde \al\in M_\al$ such that 
\beq\label{newtemp1}
\mathcal A^{u_1,u_2}_{k_1,k_2}\tilde \al=\be,\,\,\,\mathcal A^{u_1,u_2}_{k_3,k_4}\tilde \al=\gamma.
\eeq
 With definition of $f_\al$, when $\widetilde \al\in\cap_{i=1}^4 M_\al(-k_i)$, we have
 \beqa\label{temp10.6}
 f_\al(\be)=-F_a(u_1,u_2)\sqrt{\al(u_1)\al(u_2)}\Lambda^{-1}\sqrt{-w_{k_1}}\sqrt{-w_{k_2}}f_\al(\tilde \al)
\\\nonumber
 f_\al(\gamma)=-F_a(u_1,u_2)\sqrt{\al(u_1)\al(u_2)}\Lambda^{-1}\sqrt{-w_{k_1}}\sqrt{-w_{k_2}}f_\al(\tilde \al).
 \eeqa
And when $\widetilde \al\notin\cap_{i=1}^4 M_\al(-k_i)$, we have the following bound on $\left| f_\al(\be)f_\al(\gamma)\right|$,
\beq\label{temp10.62}
\left| f_\al(\be)f_\al(\gamma)\right|\leq4\al(u_1)\al(u_2)\Lambda^{-2}\prod_{i=1}^4|\sqrt{w_{k_i}}||f_\al(\tilde \al)|^2
\eeq
 \par On the other hand,  if $k_i\in P_H$ for $1\leq i\leq 4$ and 
\beq\label{defnastk1234}
\be, \gamma\in M_\al\rmand \langle \be|a^\dagger_{k_1}a^\dagger_{k_2}a_{k_3}a_{k_4}|\gamma\rangle\neq 0,
\eeq
then by the definition of $M_\al$, we have $\be(k_1)=\be(k_2)=1$ and $\gamma(k_3)=\gamma(k_4)=1$. This implies  
 \beq\label{temp10.7}
 \langle \be|a^\dagger_{k_1}a^\dagger_{k_2}a_{k_3}a_{k_4}|\gamma\rangle=1
 \eeq
 Combining \eqref{temp10.6}, \eqref{temp10.62} and \eqref{temp10.7}, we obtain that when \eqref{newtemp1} holds and  $\widetilde \al\in\cap_{i=1}^4 M_\al(-k_i)$, 
 \beq\label{temp11.14}
f_\al(\be)\overline {f_\al(\gamma)}\langle \be|a^\dagger_{k_1}a^\dagger_{k_2}a_{k_3}a_{k_4}|\gamma\rangle=F_a(u_1,u_2)^2\tilde \al(u_1)\tilde \al(u_2)\Lambda^{-2}\prod_{i=1}^4\sqrt{-w_{k_i}}|f_\al(\tilde\al)|^2
 \eeq
 When $\widetilde \al\notin\cap_{i=1}^4 M_\al(-k_i)$, using \eqref{boundspsiuroughPH}, we have 
 \beq\label{temp11.15}
 \sum_{\widetilde \al\notin\cap_{i=1}^4 M_\al(-k_i)}\left|f_\al(\be)\overline {f_\al(\gamma)}\langle \be|a^\dagger_{k_1}a^\dagger_{k_2}a_{k_3}a_{k_4}|\gamma\rangle\right|\leq \const \rho^{3/2}\al(u_1)\al(u_2)\Lambda^{-2}
 \eeq
Combining  \eqref{temp11.14} and \eqref{temp11.15}, we can see 
\beqa\label{temp10.10}
 &&A_{u_1,u_2, k_1,k_2,k_3,k_4}+O(\rho^{3/2})\al(u_1)\al(u_2)\Lambda^{-2}\\\nonumber
 =&&F_a(u_1,u_2)^2\Lambda^{-2}\prod_{i=1}^4\sqrt{-w_{k_i}}\sum_{\tilde \al\in A}\tilde \al(u_1)\tilde \al(u_2)|f(\tilde\al)|^2
 \eeqa
Where $A$ is defined as the set  
$$A\equiv\{\widetilde \al\in M_\al: \mathcal A^{u_1,u_2}_{k_1,k_2}\tilde \al=\be\in M_\al,\; \mathcal A^{u_1,u_2}_{k_3,k_4}\tilde \al=\gamma\in M_\al, \;\widetilde \al\in\cap_{i=1}^4 M_\al(-k_i)\}$$
Since $u_1,u_2\in P_L$, when $\widetilde \al \in A $, 
\beq\label{newtemp2}\widetilde \al(u_i)=\al(u_i)\,\,\,(i=1,2).
\eeq Furthermore,
 using the results in Lemma \ref{boundspsiBox}, we have that $\sum_{\tilde \al\in A}|f(\tilde\al)|^2$ bounded by 
 \beq\label{temp10.12}
 1\leq \sum_{\tilde \al\in A}|f(\tilde\al)|^2\leq 1-O(\rho^{1/6})
 \eeq
 On the other hand, using \eqref{sqrtk12k1}, with the fact $|k_1+k_2|=|k_3+k_4|\leq \rho^{1/3}\rho^{-\eta} $, one can bound the $\prod_{i=1}^4\sqrt{-w_{k_i}}$ in \eqref{temp10.10} as follows
 \beq\label{temp10.13}
\left|\prod_{i=1}^4\sqrt{-w_{k_i}}-w_{k_1}w_{k_3}\right|\leq O(\rho^{1/4-\eta})
\eeq
Inserting \eqref{newtemp2}, \eqref{temp10.13} and \eqref{temp10.12} into \eqref{temp10.10}, we arrive at the desired result \eqref{boundauukkkk}.  
\par Similarly, using the bounds on $Q_\al(0)$ and $Q_\al(0,0)$ in \eqref{boundq0} and \eqref{boundq00}, one can prove \eqref{boundauukkkk} when one of $u_i$ belongs to $ P_0$ or both of them belong to $P_0$. 
\end{proof}
With \eqref{boundauukkkk}, summing up $k_1,k_3, u_1,u_2$, one can easily obtain the desired result \eqref{resultHHHmain}.
\end{proof}
\subsection{Proof of Lemma \ref{HHHerror}}
\begin{proof} As in \cite{YY},  to estimate the error term of the interaction of particles with high momenta, we need to use a new tool. We start with defining the set $M_\al(\widetilde\al,s,\{v_1,\cdots,v_t\})$.  Let $v_1,\cdots, v_t\in P_L$ and being in different small boxes $B_L$, i.e., 
\beq
B_L(v_i)\neq B_L(v_j),\rmfor \;i\neq j.
\eeq
 For non-negative integers $s,t$ satisfying $s+t\in 2\N$ and $\widetilde \al\in M_\al$, define 
\beq\label{elementform2}
 M(\widetilde\al,s,\{v_1,\cdots,v_t\})\equiv \cup_{m}\left\{\be \in M_\al: \be=\prod_{i=m+1}^{(s+t)/2}\mathcal A^{u_{2i-1}, u_{2i}}_{p_{2i-1}, {p_{2i}}}\prod_{i=1}^{m}\mathcal A^{u_{2i-1}, u_{2i}}_{p_{2i-1}, {p_{2i}}}\,\,\widetilde\al\right\}
\eeq
where the $u_i$'s $\in P_\wL$ and $p_i$'s $\in P_H$ such that
\begin{enumerate}
\item 
$u_i=0$ for $i\leq 2m$.
\item $\{u_i$, $ 2m+1\leq i\leq s+t\}$ is a permutation of $s-2m$ zeros and $\{v_1,\cdots,v_t\}$.
\item for any fixed $2m+1\leq j\leq s+t$, $\widetilde \al\in M_\al(-p_j)$, i.e., $\widetilde \al(-p_j)=\al(-p_j)$. 
\item  $p_j\neq -p_i$ for any $2m+1\leq j\leq s+t$ and $1\leq i\leq s+t$.
\end{enumerate}
\par We note: for any $u_i$'s and $p_i$'s satisfying these four conditions, 
 one can easily  check that
\beq
\prod_{i\in A}\mathcal A^{u_{2i-1}, u_{2i}}_{p_{2i-1}, {p_{2i}}}\,\,\widetilde\al\in M_\al.
\eeq
holds for any $A\subset\{1,\cdots (s+t)/2\}$
\par
By this definition, if \eqref{defnastk1234} holds, then $\be(u)=\gamma(u)$ for any $u\in P_\wL$, then there at least  exists one $M_\al(\widetilde\al, s, \{v_i,1\leq i\leq t\})$ such that 
\beq\label{begaalsv1vt}
\be\rmand \gamma\in M_\al(\widetilde \al, s, \{v_1,\cdots,v_t\})
\eeq
E.g. Using Lemma \ref{BpMal}, we can see \eqref{begaalsv1vt} holds when we choose $\widetilde \al=\al$, $\{v_1,\cdots,v_t\}=P_L(\beta,\al)=P_L(\gamma,\al)$. 
\par Furthermore, with $M_\al(\widetilde \al,s,\{v_1,\cdots,v_t\})$, we define $N(\widetilde \al,s,\{v_1,\cdots,v_t\})$ as the set of the pairs $(\beta,\gamma)$ such that 
\begin{enumerate}
	\item $\be,\,\,\,\gamma\in M_\al(\tilde \al,s,\{v_1,\cdots,v_t\})$
	\item there exist $k_i, 1\leq i\leq 4$ satisfying \eqref{defnastk1234} but
	\beq(\be,\gamma)\notin M_\al(k_1,k_2,k_3,k_4).\eeq
	 Here $M_\al(k_1,k_2,k_3,k_4)$ is defined in \eqref{defmalkkkk}
\item for any other  $\widetilde\al\,',s',\{v'_1,\cdots,v'_{t'}\}$, if  $\be,\,\,\,\gamma\in M_\al(\widetilde\al\,',s',\{v'_1,\cdots,v'_{t'}\})$, then
\beq\label{stst}
s+t\leq s'+t' 
\eeq
\end{enumerate}
We assume \eqref{begaalsv1vt} and \eqref{defnastk1234} holds. Clearly,  $s+t=2$ or $ t=0$ implies that  $(\be,\gamma)\in M_\al(k_1,k_2,k_3,k_4)$. 
Hence if $N(\widetilde\al,s,\{v_1,\cdots,v_t\}) $ is not an empty set then 
\beq\label{stgeq2}
s+t\geq 4, \;  \text { and } \;
t\geq 1
\eeq 

By definition of $N(\widetilde \al,s,\{v_1,\cdots,v_t\})$ and \eqref{temp10.7}, we can bound the left side of \eqref{HHHerror1} as follows ($k_i\neq k_j$ for $i\neq j$)
\beqa\label{temp8.173}
&&\sum_{k_i\in P_H}\sum_{\be,\gamma \notin M_\al(k_1,k_2,k_3,k_4)}\!\!\!\!\!\!\!\!V_0\Lambda^{-1}\left|\overline{f_\al(\be)}f_\al(\gamma)\left\langle \be |a^\dagger_{k_1}a^\dagger_{k_2}a_{k_3}a_{k_4}|\gamma\right\rangle\right|\\\nonumber
\leq &&\sum_{\widetilde \al, s, \{v_1\cdots v_t\}}\!\!\!\!V_0\Lambda^{-1}\left|N(\widetilde \al,s,\{v_1,\cdots,v_t\})\right|\max_{\be,\gamma\in M_\al(\widetilde\al,s,\{v_1,\cdots,v_t\})}{\left|\overline{f_\al(\be)}f_\al(\gamma)\right|}, 
\eeqa
where $\left|N(\widetilde \al,s,\{v_1,\cdots,v_t\})\right|$ is the number of the elements in this set.
 When  \eqref{begaalsv1vt} holds, the definition of $f_\al$ implies, 
\beq\nonumber
|\overline{f_\al(\be)}f_\al(\gamma)|\leq \const^{t+s}\left|\frac{\al(0)}{|\Lambda|}\right|^{ s} \left|\frac {\rho^{-3\eta}}{|\Lambda|}\right|^{t }\max_{k\in P_H}\{|w_k|\}^{s+t}|f_\al(\tilde\al)|^2
\eeq
Here we used $m_c\leq \rho^{-3\eta}$. Again with the facts $|w_p|\leq 4\pi a |p|^{-2}$ and $\al(0)\leq N$, we obtain
\beq
|\overline{f_\al(\be)}f_\al(\gamma)|
\leq \const^{t+s}(\rho^{1-2\eta})^{s}(\rho^{-5\eta})^t|\Lambda|^{-t} |f_\al(\tilde\al)|^2
\eeq
Therefore, the r.h.s of \eqref{temp8.173} is bounded by 
\beq\label{temp8.176}
\eqref{temp8.173}\le  \sum_{\widetilde \al , s, \{v_1\cdots v_t\}}\left|N(\widetilde \al ,s,\{v_1,\cdots,v_t\})\right|\rho^s(\rho^{-6\eta})^{t+s}|\Lambda|^{-t-1} |f(\widetilde \al )|^2
\eeq
Define $N(\widetilde \al ,s,t)$ and $ N(s,t)$ by 
\beq
 N(\widetilde \al , s,t)\equiv \max_{ \{v_1,\cdots,v_t\}}\left\{\left| N(\widetilde \al , s,\{v_1,\cdots,v_t\})\right|\right\}
\eeq
\beq
 N(s,t)\equiv \max_{\widetilde \al }\left\{ N(\widetilde \al , s,t)\right\}
\eeq
With the notations $N(\widetilde \al ,s,t)$ and $ N(s,t)$, we can bound \eqref{temp8.176} by 
\beqa\nonumber
\eqref{temp8.173}\leq \eqref{temp8.176} &&\leq \sum_{\widetilde \al , s,t}|f(\widetilde \al )|^2\sum_{\{v_1\cdots v_t\}} N(\widetilde \al , s,t)\rho^s(\rho^{-6\eta})^{t+s}|\Lambda|^{-t-1}\\
&&\leq \sum_{s,t}\sum_{\{v_1\cdots v_t\}} N( s,t)\rho^s(\rho^{-6\eta})^{t+s}|\Lambda|^{-t-1}
\eeqa
For fixed $t$, the total number of sets $\{v_1\cdots v_t, v_i \in P_L\}$ is bounded by
$$
\sum_{\{v_1\cdots v_t\}}1\leq (\Lambda\rho\,\eta_L^{-3})^t(t!)^{-1}\leq (\rho^{1-3\eta})^{t}|\Lambda|^{t}(t!)^{-1}
$$
On the other hand,   $t$ is bounded above by the total number of $B_L$'s (the sides of $B_L$'s are about $\rho^{3\kappa_L}$ )in $P_L$, i.e.,   
\beq\label{ineqt}
t\leq |P_L|/\max_i\{|B_L^i|\}\leq \const\rho^{1-3\eta-3\kappa_L},
\eeq
 where $|P_L|$ and $|B_L^i|$ are the volumes of $P_L$ and the small box $B_L^i$'s. Together with  \eqref{stgeq2}, we  bound the r.h.s of \eqref{temp8.173} as follows, 
\beq\label{temp8.178}
\eqref{temp8.173}\leq \sum_{t=1}^{\rho^{1-4\eta-3\kappa_L}}\sum_{s: s+t\geq 4} N( s,t)(\rho^{1-9\eta})^{s+t}|\Lambda|^{-1}(t!)^{-1}
\eeq
We claim that $ N( s,t)$ is bounded with the following lemma, which will be proved in next subsection. 
\begin{lem}\label{lem11.1}
For any $N(\al, s,\{v_1,\cdots,v_t\})$, $s+t\geq 4$ and $t\geq 1$, we have 
\beqa\label{0.025}
\left|N(\al, s,\{v_1,\cdots,v_t\})\right|\leq t\,!\,t^{(\frac {3t}4)}|\Lambda|^{\frac{s+t}4+1} (\rho^{-\eta})^{t+s}\\\nonumber
\eeqa
\end{lem}
Combining this Lemma  with \eqref{temp8.178}, we obtain
\beqa
r.h.s\,\,\,{\rm of }\,\,\, \eqref{temp8.173}&&\leq\sum_{t=1}^{\rho^{1-4\eta-3\kappa_L}}\sum_{s: s+t\geq 4} (\rho^{1-10\eta})^{s+t}t^{(\frac {3t}4)}|\Lambda|^{\frac{s+t}4}\\\nonumber
&&=\sum_{t=1}^{\rho^{1-4\eta-3\kappa_L}}\sum_{s: s+t\geq 4}(\rho^{1-10\eta}\Lambda^{1/4})^s (\rho^{1-10\eta}t^{3/4}\Lambda^{1/4})^{t}
\eeqa
With the $\Lambda$ we chose,  $\rho^{1-10\eta}\Lambda^{1/4}$ is much less than one. Using  the assumption $\kappa_L\leq 1/2$, we have $\rho^{1-10\eta}t^{3/4}\Lambda^{1/4}\ll1$. Therefore, we arrive at the desired result:
\beq
\eqref{temp8.173}\leq O(1)\ll\rho^2\Lambda
\eeq
\end{proof}
\subsection {Proof of Lemma \ref{lem11.1}}
We now prove Lemma \ref{lem11.1}.
\begin{proof}
Since  $(\beta,\gamma)\in N(\widetilde\al, s, \{v_1,\cdots, v_t\})$, we can express them as in the r.h.s. of  \eqref{elementform2},
\beq\label{112}
\be =\prod_{i=1}^{(s+t)/2}\mathcal A^{u_{2i-1}, u_{2i}}_{q_{2i-1}, {q_{2i}}}\,\,\widetilde\al,  \; \gamma=\prod_{i=1}^{(s+t)/2}\mathcal A^{\tilde u_{2i-1}, \tilde u_{2i}}_{\tilde q_{2i-1}, {\tilde q_{2i}}}\,\,\widetilde\al
\eeq
Here $u,\widetilde u\,$'s belong to $ P_\wL$ and $q, \tilde q$'s belong to $P_H$. We note that for any $1\leq i\leq (s+t)/2$, we have
 \beq\label{temp11.38}
 \{q_{2i-1}, {q_{2i}}\}\neq \{k_1,k_2\}\rmand \{\widetilde q_{2i-1}, \widetilde q_{2i}\}\neq \{k_3,k_4\},
 \eeq
otherwise $(\be,\gamma)\in M_\al(k_1,k_2,k_3,k_4)$, which contradicts with the assumption that $(\beta,\gamma)\in N(\widetilde\al, s, \{v_1,\cdots, v_t\})$.
 \par {F}rom \eqref{defnastk1234}, one can see that the sets $\{q_1,\cdots,\,q_{2s+2t}\}$ is very close to  $\{\widetilde q_1,\cdots,\,\widetilde q_{2s+2t}\}$, i.e., 
 \beq
 \{q_1,\cdots,\,q_{2s+2t}\}-\{k_1,k_2\}=\{\widetilde q_1,\cdots,\,\widetilde q_{2s+2t}\}-\{k_3,k_4\}
 \eeq

Denote the common elements in sets $\{q_i\}$ and $\{\widetilde q_i\}$ by  $p_1$, $p_2,$ $\cdots$, $p_{s+t-2}$. Then 
we have  \beq
 \{q_i\}=\{k_1,\,\,\,k_2,\,\,\,p_1,\,\,\,p_2,\,\,\,\cdots, \,\,\,p_{s+t-2}\}
 \eeq
 \beq
 \{\widetilde q_i\}=\{k_3,\,\,\,k_4,\,\,\,p_1,\,\,\,p_2,\,\,\,\cdots, \,\,\,p_{s+t-2}\}
 \eeq
We now construct a graph with vertices  $ \{ k_1, k_2, k_3, k_4, p_i, 1\le i \le s+t -2 \} $.  The edges of the graphs are 
$\beta$ edges $(q_{2i-1},q_{2i}), 1\le i \le  (s+t)/2$ and $\gamma$ edges $(\widetilde q_{2j-1},\widetilde q_{2j}), 1\le i \le (s+t)/2$.  
{F}rom  \eqref{defnastk1234}, we know each $k_i$($1\leq i\leq 4$) touches one edge and each $p_i$($1\leq i\leq s+t-2$) touches two edges. Hence  the graph can be decomposed into two chains and loops. Thus there exist
 $l$, $m_i\in \Z$ and $0<m_1<m_2<......<m_l=s+t$ such that 
%
\beqa\label{temp4.178}
&&{\rm chains}\left\{
\begin{array}{c}
k_1\longleftrightarrow p_1\longleftrightarrow p_2\longleftrightarrow p_3\cdots p_{2m_1-1}\longleftrightarrow k_4(\rmor\,k_2)\\
k_3\longleftrightarrow p_{2m_1}\longleftrightarrow p_{2m_1+1}\cdots p_{2m_2-2}\longleftrightarrow k_2(\rmor k_4)
\end{array}\right.
\\\nonumber
&&{\rm loops}
\left\{
\begin{array}{c}
p_{2m_2-1}\longleftrightarrow p_{2m_2}{\longleftrightarrow}p_{2m_2+1}\cdots p_{2(m_3)-2}\longleftrightarrow p_{2m_2-1}
\\\cdots\cdots\\\cdots\cdots
\\p_{2m_{l-1}-1}\longleftrightarrow p_{2m_{l-1}}\longleftrightarrow p_{2m_{l-1}+1}\cdots p_{2(m_l)-2}\longleftrightarrow   p_{2m_{l-1}-1}
\end{array}\right.\eeqa
Here we have relabeled the indices of $p$ and do not distinguish $\beta$ edges and $\gamma$ edges. We also disregard the obvious 
symmetry $k_1 \to k_2$ and $k_3 \to k_4$.  
Due to the condition  \eqref{stst} and the facts $P_L(\be,\al)=P_L(\gamma,\al)$ is non-trivial(Def. \ref{defmal}), the length of the loop must be $4$ or more, i.e., each loop has at least 4 edges and 4 vertices, i.e, 
\beq\label{mi-1mi}
m_{i-1}+2\leq m_i\;\rmfor \;3\leq i\leq l
\eeq
 The inequality \eqref{temp11.38} implies $m_2\geq2$. Together with $m_l=(s+t)/2$ and \eqref{mi-1mi},  we obtain
\beq\label{uboundl}
l\leq (s+t)/4+1, \quad t \ge 1.
\eeq
Without loss of generality, we assume $m_i-m_{i-1}$ is creasing with $i\geq 3$, i.e.,  for $3\leq i<j\leq l$
\beq
m_i-m_{i-1}\leq m_j-m_{j-1}
\eeq
Denote by $N(\al, s, \{v_1,\cdots,v_t\}, l, \{m_1,\cdots,m_l\})$ the set of all pairs  $(\beta,\,\,\,\gamma)$ having the graph above
and we now estimate the number of elements of this set.

Using the notions $W_i=(w_{2i-1},w_{2i})$ and $\widetilde W_i=(\tilde w_{2i-1},\tilde w_{2i})$,  we can add the information between $k_i$'s and $p_i$'s into the graph as follows
\beqa
&&k_1\stackrel{W_1}{\longleftrightarrow} p_1\stackrel{\widetilde W_1}{\longleftrightarrow} p_2\stackrel{W_2}{\longleftrightarrow} p_3\cdots p_{2m_1-1}\stackrel {\widetilde W_{m_1}}{\longleftrightarrow} k_4(\rmor k_2)\\\nonumber
&&k_3\stackrel{\widetilde W_{m_1+1}}{\longleftrightarrow} p_{2m_1}\stackrel{W_{m_1+1}}{\longleftrightarrow} p_{2m_1+1}\cdots p_{2m_2-2}\stackrel{{ W_{m_2}}}{\longleftrightarrow} k_2(\rmor k_4)
\\\nonumber
&&p_{2m_2-1}\stackrel{W_{m_2+1}}{\longleftrightarrow} p_{2m_2}{\stackrel{\widetilde W_{m_2+1}}{\longleftrightarrow}}p_{2m_2+1}\cdots p_{2(m_3)-2}\stackrel{\widetilde W_{m_3}}{\longleftrightarrow} p_{2m_2-1}
\\\nonumber
&&\cdots\\\nonumber
&&\cdots
\\\nonumber
&&p_{2m_{l-1}-1}\stackrel{W_{m_{l-1}+1}}{\longleftrightarrow} p_{2m_{l-1}}\stackrel{\widetilde W_{m_{l-1}+1}}{\longleftrightarrow} p_{2m_{l-1}+1}\cdots p_{2(m_l)-2}\stackrel{\widetilde W_{m_{l}}}{\longleftrightarrow}   p_{2m_{l-1}-1}\, , 
\eeqa
where  
$w_i$'s are the union of $s$ zero's and $\{v_1,\cdots, v_t\}$, so are $\widetilde w$'s. More specifically, if $A\stackrel{W}{\longleftrightarrow} B$  appears in the graph and $W=(C,D)$, then the operator $\mathcal A^{C,D}_{A,B}$ appears in \eqref{112}.  
Since the momentum is conserved, we have
$$A\stackrel{W_i}{\longleftrightarrow} B\Leftrightarrow A+B=w_{2i-1}+w_{2i}, $$ so as $\widetilde W$'s.
With this relation, we can see that $\beta$ and $\gamma$ is uniquely determined by  the structure of the graph, $w_i$'s, $\widetilde w_i$'s and one $k_i$ or $p_i$ for  each loop or chain. 
\par 

To bound $| N(\widetilde\al, s, \{v_1,\cdots,v_t\}, l, \{m_1,\cdots,m_l\}) |$, we note that the sum of momentum ($p_i$'s) in each loop is zero.  
Thus we can count the number of graphs  as follows. 
\begin{enumerate}
	\item  choose the positions of zeros in $\beta$ edges. The total number of choices is less than $2^{t+s}$. 
	\item   choose the positions of $v_1\cdots v_t$ in $\beta$ edges. The total number of choices is  $t!$.
	\item   choose the positions of zeros in $\gamma$ edges. The total number of choices is less than $2^{t+s}$ again. 
	\item   choose the positions of $v_1\cdots v_t$ in $\gamma$ edges.  We call a loop trivial if all 
	the momenta associated with $\gamma$  edges are zero.   The number of trivial loops is at most $s/4$ since 
	there are at least two $\gamma$ edges(4 zero's) per loop. Hence the number of non-trivial loops is at least $l-s/4$. 
	Thus we only have to fix  $v$ in at most $t-(l-s/4)$ edges and  the number of choices is at most $t^{t-l+s/4}$. 
\end{enumerate}
Thus, with the bound on $\ell $ in   (\ref{uboundl}), we obtain 
\beqa\label{alstvlm}
 &&|N(\al, s, \{v_1,\cdots,v_t\}, l, \{m_1,\cdots,m_l\})|\\\nonumber
  \leq&& (\const)^{t+s}t! t^{(t-l+s/4)}\left(\rho^{-3\eta}\Lambda\right)^{l}\\\nonumber
 \leq &&(\const)^{t+s}t! t^{(3t/4)}\left(\rho^{-3\eta}\Lambda\right)^{t/4+s/4+1}
\eeqa

At last, with   $$| N(\al, s, \{v_1,\cdots,v_t\})|=\sum_{l}\sum_{\{m_1,\cdots,m_l\}} | N(\al, s, \{v_1,\cdots,v_t\}, l, \{m_1,\cdots,m_l\}) | 
$$ and 
\beq
\sum_{l}\sum_{\{m_1,\cdots,m_l\}} 1\leq \const^{s+t},
\eeq 
we complete the proof of \eqref{0.025}.
\end{proof}

\section{Proofs of Lemmas 1, 2, 3}
\subsection{Proof of Lemma \ref{relationDP}}\label{proofrelationDP}
The proof of Lemma \ref{relationDP} is standard and only a sketch will be given.  
We first  construct an isometry between functions with periodic 
boundary condition in $\Lambda=[0, L]^3$ and functions with Dirichlet boundary condition in $ \Lambda^*=[-\ell,\,\,\,L+\ell]^3$, where $L=\rho^{-41/60}$ and $\ell=\rho^{-41/120}$. We note, by the definition of $\rho^*$ in \eqref{relationtilde}, 
\beq\label{lrlr}
|\Lambda| \rho=|\Lambda^*| \rho^*
\eeq
\par 
Denote the coordinates of $\x$ by   $\x = (x^{(1)}, x^{(2)}, x^{(3)})$. 
Let  $h(\x)$ supported on $[-\ell,\,\,\,L+\ell]^3$ be the function $h(\x) = q(x^{(1)}) q(x^{(2)}) q(x^{(3)})$ where 
\beq
q(x)=\left\{ 
\begin{array}{ll}
\cos[(x-\ell)\pi/4\ell], & |x|\leq \ell\\
1, & \ell<x<L-\ell\\
\cos[(x-(L-\ell))\pi/4\ell], & |x-L|\leq \ell\\
0, & \text{otherwise}
\end{array}\right.
\eeq
The function  $q(x)$ is symmetric w.r.t $x=L/2$. Due to the property of cosine, for any  function $\phi$ with 
the period   $L$ we have
\beq\label{proh1}
\int_{\x\in[-\ell,\,L+\ell\,]^3}|h \phi(\x)|^2d\x=\int_{\x\in[0,L]^3}|\phi(\x)|^2d\x
\eeq
Thus the map $ \phi\longrightarrow h \phi$
is an isometry: $$
L^2_{
{\rm Periodic}}\left([0,L]^3\right)\to L^2_{
{\rm Dirichlet}}\left([-\ell,L+\ell]^3\right).$$ 
Let $\chi (\x)$ be the characteristic function of the $\ell$-boundary of $[-\ell,L+\ell]^3$, i.e., 
$\chi (\x)=1$ if   $|x^{(\alpha)} |\le \ell$ for some $\alpha = 1, 2$ or $3$ where $|x^{(\alpha)} |$ 
is  the distance on the torus $[-\ell,L+\ell]^3$. 
Then standard methods yield the following estimate on the kinetic energy of  $h \phi$
\beqa\label{3.51}
&&\int_{\x\in[-\ell,\,L+\ell\,]^3}|\nabla (h  \phi)(\x)|^2\\\nonumber
\leq &&\int_{\x\in[0,\,L]^3}|\nabla \phi(\x)|^2+\const\ell^{-2}\int  \chi (\x) | \phi(\x)|^2
\eeqa

The generalization of this isometry to higher dimensions is straightforward.   Suppose    
$\Psi(\x_1,\cdots,\x_N)$ is a function with period L. Here 
\beq
N=|\Lambda| \rho=|\Lambda^*| \rho^*
\eeq
Then  for any $u \in \R^{3}$, the map 
 \beq\label{defFu}
{\cal F}^u(\Psi): = \Psi (\x_1,\cdots,\x_N) \prod_{i=1}^N h(\x_i+u)
\eeq
is an  isometry from $L^2_{
{\rm Periodic}}\left([0,L]^{3N}\right)$ to $ L^2_{
{\rm Dirichlet}}\left([-\ell-u,L+\ell-u]^{3N}\right)$. Clearly,   ${\cal F}^u$ has the property \eqref{3.51}.

The potential $V$ can be extended to be periodic by defining   $V^P(x-y) = V( [x-y]_{P})$ 
where $[x-y]_{P}$ is the difference of $x$ and $y$ as elements on the torus $[0,L]$. 
%
Since $V$ is nonnegative and has fast decay in the position space,    we have  $V(x-y) \leq  V^P(x-y)$.
{F}rom the definition of    ${\cal F}^u$, we conclude that 
\beq\nonumber
\int_{[-\ell-u,L+\ell-u]^{3N}} \!\!\!\!\!\! \!\!\!\!\!\! \!\!\!\!\!\!|{\cal F}^u(\Psi)|^2V(\x_1-\x_2)\prod_{i=1}^N d\x_i \le 
\int_{[0,\,L]^{3N}} |\Psi|^2 V^P(\x_1-\x_2)\prod_{i=1}^Nd\x_i
\eeq

Therefore,  the total energies of ${\cal F}^u(\Psi)$ and $\Psi$ are related by 
\beq\label{calFuPsi}
\left\langle H_N\right\rangle_{{\cal F}^u(\Psi)}\leq \left\langle H_N\right\rangle_{\Psi}+\const\ell^{-2}\sum_{i=1}^N\left\langle \chi (\x_i+u)\right\rangle_{\Psi}
\eeq
We note ${\cal F}^u$ is operator on \textsl{pure} states. It can be generalized to operator ${\cal G}^u$ on states as follows. For any   state $\Gamma^P$ of $N$ particles  in $[0,L]^3$ with periodic boundary condition, we define 
 \beq
{\cal G}^u(\Gamma^P): = {\cal F}^u\Gamma^P({\cal F}^u)^\dagger
\eeq
So $\Gamma^D={\cal G}^u(\Gamma^P)$ is a state of $N$ particles  in $[-\ell-u,L+\ell-u\,]^3$ with Dirichlet boundary condition. With \eqref{lrlr}, one can see
\beq
\mathcal G^u:\,\, \Gamma^P\left( \rho,\,\, \Lambda,\,\,\beta\right)\rightarrow\Gamma^D\left(\rho^*,\,\,\Lambda^*,\,\,\beta\right)
\eeq
Using \eqref{calFuPsi}, we have:
\beq\label{calGuPsi}
\Tr \,H_N\,{\cal G}^u(\Gamma^P)\leq \Tr H_N\Gamma^P+\const\ell^{-2}\sum_{i=1}^N\Tr \chi (\x_i+u)\Gamma^P
\eeq
Averaging over $u\in[0,L]^3$, we have 
\beq
L^{-3}\int\left(\Tr H_N{\cal G}^u(\Gamma^P)\right)du\leq \Tr H_N\Gamma^P+\const\ell^{-1}L^{-1}N
\eeq
So for any $\Gamma^P$ there exists at least one $u$ such that 
\beq\label{DiffHNgugg}
\Tr H_N{\cal G}^u(\Gamma^P)\leq \Tr H_N\Gamma^P+\const N(\frac{1}{\ell L})
\eeq
On the other hand, the fact ${\cal F}^u$(\eqref{defFu}) is a isometry implies that ${\cal G}^u(\Gamma^P)$ and $\Gamma^P$ have the same von-Neumann entropy, i.e., 
\beq\label{DiffSgugg}
S({\cal G}^u(\Gamma^P))=S(\Gamma^P)
\eeq 
\par Combine \eqref{DiffHNgugg} and \eqref{DiffSgugg}, we obtain $\Delta f$ the free energy difference between ${\cal G}^u(\Gamma^P)$ and $\Gamma^P$ is less than $\const N(\ell L)^{-1}$.
With the choice $L=\rho^{-41/60}$ and $\ell=\rho^{-41/120}$, the error term  is negligible to the accuracy we need in proving  Lemma \ref{relationDP}. This concludes the proof of Lemma \ref{relationDP}.
\subsection{Proof of Lemma \ref{resultPsial}}\label{proofresultPsial}
It is not easy to define (construct) $\Gamma _0$ (the state of $N$ particles) directly. We start with constructing a state $\Gamma_{\cal F} $ in Fock space. Then pick up the useful component of $\Gamma_{\cal F}$ and revise it to $\Gamma_0$.
\par First, let $B_{\cal F}$ be the standard basis of the Fock space ${\cal F}(\Lambda)$ as follows
\beq
B_{\cal F}\equiv \left\{|\al\rangle: |\al\rangle=C_\al \prod_{k\in (\frac{2\pi \Z}{L})^3} (a^\dagger_k)^{\alpha(k)} | 0 \rangle,\,\; \al(k)\in \N\cup\{0\}\right\} , 
\eeq
where $C_\al$ is a positive normalization constant.
We define a  revised 'Bose' statistics, i.e., 
\begin{enumerate}
	\item The number of the particles in single particle state $|k\rangle $ is nonzero only when  $k\in P_I\cup P_ L$.
	\item The number of the particles in single particle state $|k\rangle $, $k\in P_L\cup P_I$, must be no more than $C_k$, which will be chosen later.
\end{enumerate}
\par With the definition of $\mu$ in \eqref{defmu}, we define $\Gamma_{\cal F}$ as the grand-canonical Gibbs state in this revised 'Bose' statistics with the chemical potential $\mu(\widetilde \rho, \beta)\leq 0$ and temperature $T=
\beta^{-1}$, where 
\beq
\widetilde \rho\equiv \rho\,(1-L^{-1/2})=\rho(1-o(\rho^{1/3}))
\eeq
and $C_k$ is chosen as follows (Recall $m_c=\rho^{-3\eta}$)
\beq\label{defCk}
C_k=\left\{\begin{array}{cc}\frac{(m_c)^{1/3}}{\beta E_{k,\mu}}& k\in P_I\\m_c&k\in P_L
\end{array}\right.,
\eeq
where $E_{k,\mu}$ is defined as $k^2-\mu(\widetilde \rho,\beta)$. We note that $\beta=O(\rho^{-2/3})$ implies,
$$\beta E_{k,\mu}C_k\geq O(\rho^{-\eta}).$$

With these notations, we can write $\Gamma_{\cal F}$ as
\beq
\Gamma_{\cal F}=C\sum_{\al\in B_{\cal F}} f_\al|\al\rangle\langle\al|
\eeq
where $C$ is a constant and $f_\al$ is non-zero only when $\al(k)$ is supported on $P_I\cup P_L$ and 
\beq
\al(k)\leq C_k, \,\,\,k\in P_I\cup P_L.
\eeq
If $f_\al$ is non-zero, 
\beq
f_{\al}\equiv \exp\left(-\sum_{k}\left(k^2-\mu(\widetilde \rho,\beta)\right)\beta \al(k)\right)
=\exp\left(-\sum_{k}E_{k,\mu}\beta \al(k)\right)
\eeq

\par We claim that the state $\Gamma_{\cal F}$ in Fock space  has the following properties:
\begin{lem}\label{profockstate}
The free energy per volume of $\Gamma_{\cal F}$ is bounded above by 
\beq\label{profockstatefe}
f(\Gamma_{\cal F})\leq f_0( \rho, \beta)(1-o(\rho^{1/3}))
\eeq
In most cases, the total particle number of $\Gamma_{\cal F}$ is less than $N=\rho\Lambda $, i.e.,
\beq\label{profockstatetn}
\sum_{m=1}^N\Tr_{\mathcal H_m}\Gamma_{\cal F}^m\geq 1-\rho
\eeq
Here $\Gamma_{\cal F}^m$ is the component of $\Gamma_{\cal F}$ on $\mathcal H_m$, i.e.,
\beq
\Gamma_{\cal F}=\sum_{m=0}^{\infty}\oplus\Gamma_{\cal F}^m,\;\;\,\,\,\,\,\,\, \Gamma_{\cal F}^m: {\cal H}_m\to{\cal H}_m
\eeq
	Similarly,  in most cases, the total particle number of $\Gamma_{\cal F}$ is very close to $\min\{\rho, \rho_c\}\Lambda$, i.e., we have
\beq\label{profockstatetn2}
\sum_{\left|m-\min\{\rho, \rho_c\}\Lambda\right|\leq N\rho^{1/3}}\Tr_{\mathcal H_m}\Gamma_{\cal F}^m\geq 1-\rho
\eeq
\end{lem}
\begin{proof}{proof of Lem.\ref{profockstate}}
\par First, we prove \eqref{profockstatefe}, by the definition,  the free energy of $\Gamma_{\cal F}$ is  
\beqa\label{fegcalf}
&&\frac{-1}{\beta}\left[\sum_{k\in P_L\cup P_I}\!\!\!\!\!\log\left(\frac{e^{\beta E_{k,\mu}}-e^{-\beta E_{k,\mu}C_k}}{e^{\beta E_{k,\mu}}-1}\right)\right]
\\\nonumber
+&&\sum_{k\in P_L\cup P_I}\!\!\!\!\!\mu(\widetilde \rho, \beta)\left(\frac{1}{e^{\beta E_{k,\mu}}-1}-\sum_{k\in P_L\cup P_I}\frac{1+C_k}{e^{ \beta E_{k,\mu}(C_k+1)}-1}\right),
\eeqa
 With the definition of $P_I$ and $P_L$, adding the $k\notin P_I \cup P_L$ terms and bounding the $C_k$ terms, one can easily  check that \eqref{fegcalf} is equal to 
\beq\label{fegcalf2}
\left(\frac{-1}{\beta}\!\!\!\!\!\sum_{k\in (\frac{2\pi \Z}{L})^3, k\neq 0}\!\!\!\!\!\log\left(\frac{e^{\beta E_{k,\mu}}}{e^{\beta E_{k,\mu}}-1}\right)+\!\!\!\!\!\sum_{k\in (\frac{2\pi \Z}{L})^3, k\neq 0}\mu(\widetilde \rho, \beta)\frac{1}{e^{\beta E_{k,\mu}}-1}\right)(1+o(\rho^{1/3})),
\eeq

Then with the choice $L= \rho^{-41/60}$ and the definition of free energy $f_0$ in \eqref{deff01} and \eqref{deff02}, we have 
\beq
\eqref{fegcalf2}=f_0(\widetilde \rho, \beta)\Lambda(1+o(\rho^{1/3}))
\eeq Combining this  with   $\widetilde \rho=\rho(1+o(\rho^{1/3}))$, we obtain the desired result \eqref{profockstatefe}.
\par 
Then we prove \eqref{profockstatetn}. Let $n(k)$ denote  the number of the particles in one-particle-state $|k\rangle$.  Then $\overline{n(k)}$ the average of $n(k)$ is equal to $ \Tr \cre_k\ann_k\Gamma_{\cal F}$. By the definition, the average total number of particles of $\Gamma_{\cal F}$ is equal to 
\beq\label{tncalf}
\sum_{ k\in P_I\cup P_L }\overline{n(k)}=\sum_{k\in P_I\cup P_L}\frac{1}{e^{\beta E_{k,\mu}}-1}-\sum_{k\in P_L\cup P_I}\frac{1+C_k}{e^{ \beta E_{k,\mu}(C_k+1)}-1}
\eeq
Similarly, with $L= \rho^{-41/60}$ and $\beta E_{k,\mu}C_k\gg|\log\rho|$, one can easily prove: 
\beqa\label{temp11.20}
\eqref{tncalf}&=&\min\{\widetilde\rho, \rho_c(\beta)\}\Lambda(1+O(\rho^{-1/3}L^{-1}\log \rho))
\\\nonumber &=&\min\{\widetilde\rho, \rho_c\}\Lambda+o(N\rho^{41/120})
\eeqa
On the other hand, we are going to use Hoeffding's inequality  to estimate $\sum_{k}n(k)$. Hoeffding's inequality said, for independent $X_i$'s, if they are bounded as
\beq
a_i\leq X_i-\mathbb E(X_i)\leq b_i
\eeq
where $\mathbb E(X_i)$ is  the expected value of $X_i$, then 
\beq
P\left(\left|\sum_i  X_i-\mathbb E(\sum_i X_i)\right|>t\right)
\leq 2\exp\left(-\frac{2t^2}{\sum_i (b_i-a_i)^2}\right)
\eeq
Since $n(k)$'s are independent random variables for different $k$'s and they are bounded in \eqref{defCk}, we can use Hoeffding's inequality \cite{H} to estimate the distribution of the total particle number of $\Gamma_{\cal F}$. With $n(k)\leq C_k$ and Hoeffding's inequality \cite{H}, we obtain that  the probability of finding more than $N$ particles in $\Gamma_{\cal F}$ is bounded above by
\beq\label{Hoeff}
P\left(\sum_k n(k)>N\right)\leq 2 \exp\left\{-\frac{2\left[N-\sum_{ k }\overline{n(k)}\right]^2}{\sum_{k\in P_I\cup P_L}C_k^2}\right\}
\eeq
By the definition of $C_k$ \eqref{defCk}, the denominator of the r.h.s of \eqref{Hoeff} is bounded as :
\beq\label{temp11.21}
\sum_{k\in P_I\cup P_L}C_k^2= O(\rho^{4/3}\Lambda Lm_c^{2/3})
\eeq
On the other hand, with the fact $\rho-\widetilde\rho=\rho L^{-1/2}$ and \eqref{temp11.20},  the numerator  of the r.h.s of \eqref{Hoeff} is bounded below by
\beq\label{temp11.22}
\left[N-\sum_{ k }\overline{n(k)}\right]^2\geq O(\rho^2 L^5)
\eeq
Inserting  $L= \rho^{-41/60}$, \eqref{temp11.21} and \eqref{temp11.22} into \eqref{Hoeff}, we obtain the desired result \eqref{profockstatetn}. And \eqref{profockstatetn2} can proved similarly with \eqref{temp11.20} and \eqref{temp11.21}.
\end{proof}
\par By Lemma \ref{profockstate}, there exists $m_0\leq N$ such that 
\beq\label{relamN}
m_0\leq N, \,\,\, |m_0-\min\{\rho, \rho_c\}\Lambda|\leq \rho^{1/3}N
\eeq
and the free energy of $\Gamma_{\cal F}^{m_0}$ is less than $f_0(\rho,\beta)\Lambda(1-o(\rho^{1/3}))$. 
\par Then adding $N-m_0$ ($N=\rho\Lambda$) particles with momentum zero into the system described by $\Gamma_{\cal F}^{m_0}$, we obtain a new state $\Gamma_0$ of $N$  particles. The state $\Gamma_0$  always has $N-m_0$ particles with momentum zero.  The free energy of $\Gamma_{0}$ is also less than $f_0(\rho,\beta)\Lambda(1-o(\rho^{1/3}))$, i.e.,
\beq\label{nonintfepr}
\left|\Tr (-\Delta  \Gamma_0)+\frac{1}{\beta}S( \Gamma_0)-f_0(\rho,\beta)\right|\Lambda^{-1}\leq o(\rho^2) 
\eeq
Furthermore, by the definition of $\Gamma_{\cal F}$, $\Gamma_0$ has the form:
\beq
\Gamma_0=\sum_{\al\in M} g_\al(\rho,\beta)|\al\rangle\langle\al|,\;\; \al(0)=N-m_0\rmand \sum_{\al\in M}g_\al=1
\eeq
We note: if $\al(k)>C_k$ for some $k\in P_I\cup P_L$, then $g_\al(\rho,\beta)=0$. This property implies that  the total number of the particles with momentum in $P_I$ is $o(N)$. So we have
\beq
\sum_{\al\in M}\sum_{k\in P_I}g_\al(\rho,\beta)\al(k)\ll N.
\eeq
Together with the facts $\al(0)=N-m_0$, \eqref{relamN} and $\al(k)\leq m_c$ for $\al\in P_L$, we obtain \eqref{resultPsialequ2}. 
\par At last we prove \eqref{prowgamma1}. First with the structure of $\Gamma_0$, we have
\beqa
\Tr_{\mathcal H_{N,\Lambda}} \frac12 V\,\Gamma_0
=&&\sum_{\al\in M}g_\al(\rho,\beta)\langle\al|\frac12V|\al\rangle\\\nonumber
=&&\sum_{\al\in M}g_\al(\rho,\beta) \bigg(\sum_{k\in P_0\cup P_I\cup P_L}\frac12V_0\Lambda^{-1}(\al(k)^2-\al(k))\\\nonumber
&&+\sum_{k,k'\in P_0\cup P_I\cup P_L}^{ k\neq k'}(V_0+V_{k-k'})\Lambda^{-1}\al(k)\al(k')\bigg)
\eeqa
Using the smoothness of  $V$ and $|k|, |k'|\ll1$, we can replace $V_{k-k'}$ with $V_0$ without changing the leading term. Then with the cutoff $C_k$'s, the fact  $\al(0)=N-m$ and \eqref{relamN}, we have
\beq
\lim_{\rho\to 0}|\Tr \frac12V\Gamma_0|\rho^{-2}\Lambda^{-1}=\frac12V_0(2-[1-R[\beta]\,]^2_+)
\eeq
Combine  with \eqref{nonintfepr},  we obtain \eqref{prowgamma1}.

\subsection{Proof of Lemma \ref{entropy}}\label{proofentropy}
\begin{proof}
Since the states  $|\al\rangle$'s  $\in M$ are orthonormal, we can rewrite the entropy of  $\Gamma_0$ in lemma \ref{resultPsial} as 
\beq
S(\Gamma_0)=-\sum_{\al\in M}g_\al\log g_\al
\eeq
For $S(\Gamma)$, we define $A_\infty$ as $$A_\infty\equiv\left\|\sum_{\al\in M}|\Psi_\al\rangle\langle\Psi_\al|\right\|_{\infty}$$ and rewrite $\Gamma$ as 
\beq
\Gamma=A_\infty\sum_{\al\in M}g_\al\frac{|\Psi_\al\rangle}{\sqrt{A_\infty}}\,\frac{\langle \Psi_\al|}{\sqrt{A_\infty}}
\eeq
With the fact $\Tr \Gamma=1$, i.e., $\sum g_\al=1$, we have
\beq
S(\Gamma)=-\log A_\infty-A_\infty \Tr\left[\sum_{\al\in M}g_\al\frac{|\Psi_\al\rangle}{\sqrt{A_\infty}}\,\frac{\langle \Psi_\al|}{\sqrt{A_\infty}}\log( \sum_{\al\in M}g_\al\frac{|\Psi_\al\rangle}{\sqrt{A_\infty}}\,\frac{\langle \Psi_\al|}{\sqrt{A_\infty}})\right]
\eeq
With the concavity of the logarithm, one can easily obtain 
\beq\label{temp11.34}
S(\Gamma)\geq -\log A_\infty-\sum_{\al\in M}g_\al\log g_\al=-\log A_\infty+S(\Gamma_0)
\eeq
\par  We claim the following lemma
\begin{lem}\label{lastlemma}
\beq\label{resultlastlemma}
\lim_{\rho\to0}\left(\log\left\|\sum_{\al\in M}|\Psi_\al\rangle\langle\Psi_\al|\right\|_{\infty}\right)\frac{1}{N\rho^{1/3}}=0
\eeq
\end{lem}
Insert this lemma into \eqref{temp11.34}, we arrive at the desired result \eqref{resultentropy}.
\end{proof}
\subsubsection{Proof of Lemma \ref{lastlemma}}
\begin{proof}
With the fact: for any hermitian matrix $M=M_{ij}$,
$$\|M\|_\infty\leq \max_{i} \left\{\sum_{j}|M_{ij}|\right\},$$
we can bound $\|\sum_{\al\in M}|\Psi_\al\rangle\langle\Psi_\al|\|_\infty$ as follows (Recall $\widetilde M$ in Def. 2.)
\beqa\label{proofentropy1}
\|\sum_{\al\in M}|\Psi_\al\rangle\langle\Psi_\al|\|_\infty\leq &&\max_{\be\in\widetilde M}
\left\{\sum_{\al\in M}\sum_{\gamma\in\widetilde M}\left|\langle\be|\Psi_\al\rangle\langle\Psi_\al|\gamma\rangle\right|\right\}\\\nonumber
\leq &&\max_{\be\in\widetilde M}\left\{\sum_{\al\in M}|\langle\be|\Psi_\al\rangle|\right\}\cdot\max_{\al\in M }\left\{\sum_{\gamma\in\widetilde M}|\langle\gamma|\Psi_\al\rangle|\right\},
\eeqa
With the fact $\Psi_\al$ is the linear combination of states in $M_\al\subset\widetilde M_\al$ and $|\be\rangle$, $|\Psi_\al\rangle$  are normalized, we claim 
\beqa\label{maxbewtM}
&&\log\left( \max_{\be\in\widetilde M}\left\{\sum_{\al\in M}|\langle\be|\Psi_\al\rangle|\right\}\right)\leq \rho^{1-4\eta-3\kappa_L} \\\label{maxalM}
&&\log \left(\max_{\al\in M }\left\{\sum_{\gamma\in\widetilde M}|\langle\gamma|\Psi_\al\rangle|\right\}\right)\leq \rho^{1-4\eta-3\kappa_L}+\rho^{-4\eta-3\kappa_H} 
\eeqa
First, we prove \eqref{maxbewtM}. For any $\al\in M$ and $\be\in \widetilde M_\al$,  $|\langle\be|\Psi_\al\rangle|\neq 0$ implies $|\langle\be|\Psi_\al\rangle|\leq1$.   Then with the definition of $M$ and $\widetilde M_\al$, if $\al\in M$, $\be\in \widetilde M_\al$, we have 
\beqa\label{temp12.43}
&&\be(u)=\al(u)\rmfor u\in P_I\\\nonumber
&&\be(u)\leq \al(u) \rmfor u\in P_L\\\nonumber
&&\al(u)=0\rmfor u\in P_H
\eeqa
and for any fixed small box $B_L^{\,i}(i=1,2,\dots)$ in $P_L$, $\be(u)$ is very close to $\al(u)$, i.e.,
\beq\label{temp12.44}
\sum_{u\in B_L^{\,i}} \left|\be(u)-\al(u)\right|\leq 1
\eeq
Now let's count, for fixed $\be$, how many $\al\in M$ satisfy $\be\in \widetilde M_\al$. This number must be less than the $\al$'s satisfying \eqref{temp12.43} and \eqref{temp12.44}. By the definition of $B_L$'s, the total number of $B_L$'s is less than  $\const \rho^{1-3\eta-3\kappa_L}$. And for any $B_L^i$, $|B_L^i|$ the number of the elements in $B_L^i$ is less than $\const \rho^{3\kappa_L}\Lambda$. Therefore, for fix $\be \in \widetilde M$, the total number of $\al\in M$ satisfying $\be\in \widetilde M_\al$ is less than 
\beq
\left(\const \rho^{3\kappa_L}\Lambda\right)^{\const \rho^{1-3\eta-3\kappa_L}}
\eeq
Together with the fact $|\langle\be|\Psi_\al\rangle|\leq 1$, we proved \eqref{maxbewtM}.
\par Then we prove  \eqref{maxalM}. Similarly, using the rule 2 of Def. 3, we can count, for  fix $\al \in  M$, the total number of $\gamma\in \widetilde M$, s.t. $|\langle\gamma|\Psi_\al\rangle|\neq 0$ is less than 
\beq
\left(\const \rho^{3\kappa_L}\Lambda\right)^{\const \rho^{1-3\eta-3\kappa_L}}\left(\const \rho^{3\kappa_H}\Lambda\right)^{\const \rho^{-3\eta-3\kappa_H}},
\eeq
which implies \eqref{maxalM}. Inserting \eqref{maxbewtM} and \eqref{maxalM} into \eqref{proofentropy1}, we obtain the desired result \eqref{resultlastlemma}.
\end{proof}

\bigskip 

\section{ Appendix}
\begin{lem}\label{ffg}
For any bound, non-negative, piecewise continous function, spherically symmetric $f$ supported in unit ball, there exist $C^\infty$, non-negative  spherically symmetric function $f_{1}$, $f_2, \ldots$ supported in the ball of radius 2 such that for any $n\geq 1$, 
 \beq
 f_n-f\geq 0\,\,\,and\,\,\, \lim_{n\to \infty}\|f_n-f\|_1= 0
 \eeq
 \end{lem}
 
 \bigskip
 
 {\it Proof:}
 First, we note, for any bound, non-negative, piecewise continous function $f$ supported in unit ball, there exist non-negative,  countinous functions $\tilde f_1$, $\tilde f_2,\ldots$ supported in the ball of radius 1.5, such that
 \beq
 \tilde f_n\geq f\,\,\,{\rm and} \,\,\,\,\lim_{n\to \infty}\|\tilde f_n-f\|_1\to 0
 \eeq
 Then we  claim that for any $\tilde f_n$, there exist $C^\infty$, non-negative  spherically symmetric function $\tilde f_{nm}$ ($m=1,2,\ldots$) supported in the ball of radius 2, such that 
 \beq\label{13.3}
 \tilde f_{nm}\geq \tilde f_{n}\,\,\,{\rm and} \,\,\,\,\lim_{n\to \infty}\|\tilde f_n-\tilde f_{nm}\|_1\to 0
 \eeq
 To prove lemma \ref{ffg}, we can  choose $f_n$ as $\tilde f_{nm_n}$, where $m_n$ is defined as 
 \beq
 m_n=\min\left\{m:\,\,\,\|\tilde f_n-\tilde f_{nm}\|_1\leq \|\tilde f_n-f\|_1\right\}
 \eeq  
 It only remains to prove \eqref{13.3}. Let $g$ be a bound $C^\infty$ spherically symmetric function support in the ball of radius 2 such that  
 \beq
g\geq 0,\,\,\, \|g\|_1=1 \,\,\,and \,\,\,g(x)=g(0)>0\,\,\, for \,\,\, |x|\leq 1.5
 \eeq
 And we define $g_m$ as 
 \beq
 g_m(x)=m^3g(mx)
 \eeq
 Then $\|g_m\|_1=1$. Furthermore, for fixed $n$ 
 \beq
 \lim_{m\to \infty}\|\tilde f_n*g_m-\tilde f_n\|_\infty=0\;{\;\;\;\rm and}\;
 \lim_{m\to \infty}\|\tilde f_n*g_m-\tilde f_n\|_1=0
 \eeq
 and $\tilde f_n*g_m$ are non-negative,  $C^\infty$ spherically symmetric functions. Since $\tilde f_n$ is supported on the ball of radius $1.5$,  we can choose $\tilde f_{nm}$ as
 \beq
\tilde f_{nm}\equiv  \tilde f_n*g_m+\frac1{g(0)}\,\|\tilde f_n*g_m-\tilde f_n\|_\infty\, g
 \eeq
 and complete the proof. 
 \qed

\end{document}